\documentclass[12pt]{article}
\usepackage[utf8]{inputenc}
\usepackage{bm}
\usepackage{amsfonts}
\usepackage{amsmath}
\usepackage{amssymb}
\usepackage{amsthm}
\usepackage{bbm}
\usepackage{graphicx}
\usepackage{pdflscape}
\usepackage{tabularx}
\usepackage{array}
\usepackage{booktabs}
\usepackage[font=bf,justification=centering]{caption}
\usepackage{enumitem}
\usepackage[top=1.25in, bottom=1.25in, left=1.25in, right=1.25in]{geometry}
\usepackage{setspace}
\usepackage{siunitx}
\usepackage{titling}
\usepackage{natbib}
\usepackage[toc,page]{appendix}
\usepackage[hidelinks,bookmarks]{hyperref}

\usepackage{bibunits}

\defaultbibliography{bibliography}        
\defaultbibliographystyle{ecta}

\newtheorem{algorithm}{Algorithm}
\theoremstyle{definition}
\newtheorem{assumption}{Assumption}
\theoremstyle{definition}

\theoremstyle{plain}

\theoremstyle{plain}

\newtheorem{proposition}{Proposition}

\hypersetup{pdfstartview={XYZ null null 1.00}} 
\allowdisplaybreaks

\newcolumntype{Y}{>{\centering\arraybackslash}X}

\title{Fast Posterior Sampling in Tightly Identified SVARs Using `Soft' Sign Restrictions\thanks{We thank Joshua Chan, Thorsten Drautzburg, Toru Kitagawa, James Morley, Yong Song, Benjamin Wong, three anonymous referees and the Co-Editor, James Hamilton, for valuable comments. We also thank participants at the 2023 Continuing Education in Macroeconometrics workshop, the 2023 Time Series and Forecasting Symposium, the 2024 Econometric Society Australasia Meeting and the 2025 World Congress of the Econometric Society. The views expressed in this paper are those of the authors and should not be attributed to the Reserve Bank of Australia. Any errors are the sole responsibility of the authors.}}
\author{Matthew Read\thanks{Economic Research Department, Reserve Bank of Australia. email: readm@rba.gov.au}\\
Dan Zhu\thanks{Department of Econometrics and Business Statistics, Monash University. email: dan.zhu@monash.edu.au}}
\date{\today}

\begin{document}

\begin{bibunit}

\maketitle

\begin{abstract}
	We propose algorithms for conducting Bayesian inference in structural vector autoregressions identified using sign restrictions. The key feature of our approach is a sampling step based on `soft' sign restrictions. This step draws from a target density that smoothly penalises parameter values that violate the restrictions, facilitating the use of computationally efficient Markov chain Monte Carlo sampling algorithms. An importance-sampling step yields draws conditional on the `hard' sign restrictions. Relative to standard accept-reject sampling, the method substantially speeds up sampling when identification is tight. It also facilitates implementing prior-robust Bayesian methods. We illustrate the broad applicability of the approach in an oil-market model identified using a rich set of sign, elasticity and narrative restrictions.
\end{abstract}


\section{Introduction}
\label{sec:introduction}
Sign-restricted structural vector autoregressions (SVARs) are used extensively to estimate the effects of macroeconomic shocks. Researchers using sign restrictions typically work with the SVAR's `orthogonal reduced-form' parameterisation and conduct Bayesian inference under a uniform (Haar) prior on the orthonormal (rotation) matrix, $\mathbf{Q}$.\footnote{See the references in \citet{Baumeister_Hamilton_2018} for many examples. There is debate about the use of this prior. In set-identified models, a component of the prior is never updated (\citealt{Poirier_1998}). This can result in posterior inference being asymptotically sensitive to the choice of prior and Bayesian credible sets lying strictly within frequentist confidence sets (\citealt{Moon_Schorfheide_2012,Baumeister_Hamilton_2024}). These concerns tend to be negligible when identification is tight, which is the setting we focus on.} Posterior sampling in this setting is almost universally implemented using accept-reject sampling (e.g. \citealt{Uhlig_2005,Rubio-Ramirez_Waggoner_Zha_2010,Arias_Rubio-Ramirez_Waggoner_2018}).

Sign restrictions generate an identified set for $\mathbf{Q}$, which is the set of observationally equivalent parameter values. Accept-reject sampling can be computationally demanding when the identified set has small measure -- that is, when identification is tight. For example, in a model from \citet{Kilian_Murphy_2014}, \citet{Baumeister_Hamilton_2024} report that only 16 out of 5~million draws satisfy the identifying restrictions. In such cases, accurately approximating the posterior can require a very large number of proposal draws. This bottleneck is likely to intensify as applied work moves towards richer sets of identifying restrictions (e.g. \citealt{Inoue_Kilian_2026}).\footnote{Rich sets of identifying restrictions can mitigate the problem of `masquerading' shocks, where linear combinations of shocks satisfy the sign restrictions (\citealt{Wolf_2020,Wolf_2022}).}

This paper develops an approach to sampling $\mathbf{Q}$ that addresses the bottleneck associated with accept-reject sampling. Its key feature is an initial sampling step based on `soft' sign restrictions. The approach involves: 1) specifying a target distribution that smoothly penalises parameter values violating the sign restrictions; 2) using a Markov chain Monte Carlo (MCMC) algorithm -- the slice sampler (e.g. \citealt{Neal_2003}) -- to draw from this target distribution; and 3) importance sampling to obtain uniform draws of $\mathbf{Q}$ conditional on the (hard) sign restrictions. We emphasise that we do not conduct inference under a relaxed set of restrictions; the soft sign restrictions are a computational device to facilitate sampling under the hard restrictions.\footnote{For approaches to conducting Bayesian inference under non-dogmatic sign restrictions, see \citet{Baumeister_Hamilton_2018,Baumeister_Hamilton_2019}.}

Our approach delivers efficiency gains for two reasons. First, we sample from a smooth density that tends to assign higher probability to the identified set and smaller but nonzero probability outside the identified set. This means the sampler, when initialised outside the identified set, will tend to drift from its initial point towards the identified set. Second, once the sampler has moved within the identified set, it tends to remain there. By initially sampling from a computationally convenient approximation to the target density and then using importance sampling to obtain draws from the hard-truncated distribution, our approach is analogous to tempering, which is often used in Sequential Monte Carlo (e.g. \citealt{Herbst_Schorfheide_2015}).

Existing approaches to uniform sampling of $\mathbf{Q}$ under sign restrictions can in some cases be more efficient than accept-reject sampling, but are only applicable under certain classes of restrictions. \citet{Amir-Ahmadi_Drautzburg_2021} develop a Gibbs sampler that draws from a uniform distribution over the identified set for $\mathbf{Q}$. However, the sampler is not applicable when the sign restrictions constrain all columns of $\mathbf{Q}$ (e.g. when there are restrictions on the impulse responses to all shocks) or when each restriction does not linearly constrain a single column of $\mathbf{Q}$.\footnote{\citet{Read_2022} extends this sampler to allow for zero restrictions.}

\citet{Chan_Matthes_Yu_2025} exploit invariance of the uniform distribution for $\mathbf{Q}$ to column permutations and sign flips to efficiently generate draws that satisfy sign and ranking restrictions on impact responses. Their method can be combined with an accept-reject step to impose additional restrictions, though this will be costly if these restrictions substantially tighten identification. The gains from their method are largest in high-dimensional SVARs, as there are many possible permutations, while the gains are smaller in the lower-dimensional SVARs that we focus on.\footnote{\citet{Hou_2024} proposes an MCMC algorithm for posterior sampling under linear equality and inequality restrictions on impact impulse responses, with additional inequality restrictions implementable via an accept-reject step. The posterior corresponds to an independent Gaussian prior over the columns of the impact impulse-response matrix, rather than the uniform prior for $\mathbf{Q}$.}

In contrast with these existing approaches, our sampler can be directly applied when the identifying restrictions are nonlinear and/or constrain all columns of $\mathbf{Q}$, and can thus handle a wide variety of identifying restrictions, provided these take the form of inequalities. These include restrictions on: impulse responses (e.g. \citealt{Uhlig_2005}); structural coefficients (e.g. \citealt{Arias_Caldara_Rubio-Ramirez_2019}); ratios of impulse responses and elasticities (e.g. \citealt{Kilian_Murphy_2012,Baumeister_Hamilton_2024}); and forecast error variance decompositions (e.g. \citealt{Volipcella_2021}). They are also applicable under: shape or ranking restrictions (e.g. \citealt{Amir-Ahmadi_Drautzburg_2021}); narrative restrictions (e.g. \citealt{Ludvigson_Ma_Ng_2018,Ludvigson_Ma_Ng_2021,Antolin-Diaz_Rubio-Ramirez_2018}); and restrictions on the relationship between proxies and structural shocks (e.g. \citealt{Arias_Rubio-Ramirez_Waggoner_2021,Giacomini_Kitagawa_Read_2022b,Braun_Bruggemann_2022}).

In contemporaneous work, \cite{Arias_Rubio-Ramirez_Shin_2026} develop a Gibbs sampler to sample from the posterior of the orthogonal reduced-form parameters, focusing on SVARs with a large number of variables. Their sampler requires iteratively sampling from conditional densities that are truncated by the sign restrictions using the elliptical slice sampler from \cite{Murray_Adams_MacKay_2010}. Soft sign restrictions could in principle be exploited to more efficiently sample from these conditional densities when the sign restrictions substantially truncate their support.

Accept-reject sampling also causes bottlenecks when using prior-robust Bayesian methods to assess or eliminate prior sensitivity (e.g. \citealt{Giacomini_Kitagawa_2021,Giacomini_Kitagawa_Read_2026}). Implementing these methods requires calculating bounds of identified sets at every reduced-form parameter draw. This can be achieved by repeatedly drawing $\mathbf{Q}$ from inside its identified set and computing the minimum and maximum of the parameter of interest (e.g. an impulse response) over the draws. A large number of draws may be required to accurately approximate identified sets (e.g. \citealt{Montiel-Olea_Nesbit_2021}). Our sampler facilitates applying prior-robust Bayesian methods, particularly when identification is tight.

The idea of `softening' restrictions has been used in other contexts. In the context of sampling, \citet{Souris_Bhattacharya_Pati_2019} develop a Gibbs sampler to draw from a smooth approximation of a truncated multivariate normal distribution subject to linear constraints. Our proposal also involves sampling from a smooth approximation of a truncated normal distribution, but the general set of constraints that we consider may be nonlinear. In operations research, kernel smoothing methods have been used to compute derivatives of simulated outcomes when sample paths are discontinuous (e.g. \citealt{Liu_Hong_2011,Bruins_etal_2018}).

We demonstrate the utility of our approach in two main settings. First, we use a bivariate model to illustrate the mechanics of the method and compare its efficiency with accept-reject sampling. Our approach performs particularly well when the identified set for $\mathbf{Q}$ has small measure under the uniform prior, which occurs when identification is tight. We also show that the approach continues to sample effectively even when the identified set consists of disconnected regions.

Second, we revisit the oil-market model in \citet{Antolin-Diaz_Rubio-Ramirez_2018}, which builds on \citet{Kilian_2009} and \citet{Kilian_Murphy_2012}.\footnote{Closely related variants of this model are considered in several other papers (e.g. \citealt{Baumeister_Peersman_2013,Lutkepohl_Netsunajev_2014,Baumeister_Hamilton_2019,Bacchiocchi_etal_2024,Carriero_Marcellino_Tornese_2024,Hoesch_Lee_Mesters_2024}).} This model imposes a rich set of sign, elasticity and narrative restrictions, which nonlinearly constrain all columns of $\mathbf{Q}$. Our sampler is roughly an order of magnitude more efficient than accept-reject sampling in this application. We also demonstrate the utility of our approach in conducting prior-robust Bayesian inference, which would be extremely computationally burdensome when implemented via accept-reject sampling. We find that inferences about the effects of shocks in the oil market obtained under this rich set of identifying restrictions are broadly robust to the choice of conditional prior for $\mathbf{Q}$. Importantly, this apparent robustness is unlikely to be an artifact of numerical approximation error given the large number of draws used to approximate identified sets. We briefly outline an additional empirical application -- a model of US monetary policy from \citet{Antolin-Diaz_Rubio-Ramirez_2018} -- which indicates that the favourable performance of our sampler persists in a larger model.

The remainder of the paper is structured as follows. Section~\ref{sec:framework} describes the SVAR, the identifying restrictions and the sampling problem. Section~\ref{sec:algorithms} describes accept-reject sampling and introduces our sampler based on soft sign restrictions. Section~\ref{sec:numericalexamples} illustrates our sampler in a simple example. Section~\ref{sec:empirical} examines the performance of the sampler empirically. Section~\ref{sec:conclusion} concludes. Appendix~\ref{app:proofs} contains proofs. A Supplemental Appendix contains additional material. 

\medskip

\textbf{Notation.} $\mathbf{e}_{i,n}$ is the $i$th column of the $n\times n$ identity matrix, $\mathbf{I}_{n}$. $\mathbf{0}_{n\times m}$ is an $n\times m$ matrix of zeros. For a $n\times m$ matrix $\mathbf{X}$, $\mathrm{vec}(\mathbf{X})$ is the vectorisation operator, which stacks the elements of $\mathbf{X}$ into an $nm\times 1$ vector. If $\mathbf{X}$ is $n\times n$, $\mathrm{vech}(\mathbf{X})$ is the half-vectorisation, which stacks the elements lying on or below the diagonal into an $n(n+1)/2\times 1$ vector. $\mathbbm{1}(.)$ is the indicator function.

\section{Framework}
\label{sec:framework}
This section describes the SVAR and its orthogonal reduced-form parameterisation, the identifying restrictions considered and the sampling problem.

\subsection{SVAR and orthogonal reduced form}
\label{subsec:SVAR}
Let $\mathbf{y}_t=\left(y_{1t},\ldots,y_{nt}\right)'$ be an $n\times 1$ vector and consider the SVAR($p$):
\begin{equation}
	\mathbf{A}_0 \mathbf{y}_t = \mathbf{A}_+ \mathbf{x}_t + \bm{\varepsilon}_t,
\end{equation} 
where $\mathbf{A}_0$ is an invertible $n\times n$ matrix and $\mathbf{x}_t = \left(\mathbf{y}_{t-1}',\ldots,\mathbf{y}_{t-p}'\right)'$.\footnote{In our illustrations and empirical exercises, we normalise the diagonal elements of $\mathbf{A}_0$ to be nonnegative, but this is not a requirement of our methods.} Conditional on past information, $\bm{\varepsilon}_{t} \sim N(\mathbf{0}_{n\times 1},\mathbf{I}_{n})$. The orthogonal reduced form is
\begin{equation}
	\mathbf{y}_t = \mathbf{B}\mathbf{x}_t + \bm{\Sigma}_{tr}\mathbf{Q}\bm{\varepsilon}_t,
\end{equation}
where: $\mathbf{B} = \left(\mathbf{B}_1,\ldots,\mathbf{B}_p\right)=\mathbf{A}_0^{-1}\mathbf{A}_+$ are the reduced-form coefficients; $\bm{\Sigma}_{tr}$ is the lower-triangular Cholesky factor of the variance-covariance matrix of the reduced-form VAR innovations, $\bm{\Sigma}=\mathbb{E}\left(\mathbf{u}_t\mathbf{u}_t'\right) = \mathbf{A}_{0}^{-1}\left(\mathbf{A}_{0}^{-1}\right)'$ with $\mathbf{u}_{t} = \mathbf{y}_{t} - \mathbf{B}\mathbf{x}_t$; and $\mathbf{Q} \in \mathcal{O}(n)$ is an $n\times n$ orthonormal matrix with $\mathcal{O}(n)$ the space of all such matrices. Denote the reduced-form parameters by $\bm{\phi} = \left(\mathrm{vec}\left(\mathbf{B}\right)',\mathrm{vech}(\bm{\Sigma}_{tr})'\right)' \in \bm{\Phi}$.

Let $\mathbf{C}_h$ be defined recursively by $\mathbf{C}_h=\sum_{l=1}^{\mathrm{min}\{h,p\}}\mathbf{B}_l\mathbf{C}_{h-l}$ for $h\geq 1$ with $\mathbf{C}_0 = \mathbf{I}_n$. Element $(i,j)$ of $\mathbf{C}_h\bm{\Sigma}_{tr}\mathbf{Q}$ is the horizon-$h$ impulse response of variable $i$ to structural shock $j$, denoted by $\eta_{ijh}(\bm{\phi},\mathbf{Q})=\mathbf{c}_{ih}'(\bm{\phi})\mathbf{q}_j$, where $\mathbf{c}_{ih}'(\bm{\phi}) = \mathbf{e}_{i,n}'\mathbf{C}_{h}\bm{\Sigma}_{tr}$ is row $i$ of $\mathbf{C}_{h}\bm{\Sigma}_{tr}$ and $\mathbf{q}_{j} = \mathbf{Q}\mathbf{e}_{j,n}$ is column $j$ of $\mathbf{Q}$.

\subsection{Identifying restrictions}
\label{subsec:identification}
Identifying restrictions on functions of the structural parameters are equivalent to restrictions on $\mathbf{Q}$. We allow for a wide variety of inequality restrictions (which we refer to as `sign restrictions'), including restrictions on:
\begin{itemize}
	\item \textbf{Impulse responses.} The restriction $\eta_{ijh}(\bm{\phi},\mathbf{Q}) \geq 0$ is equivalent to  $\mathbf{c}_{ih}'(\bm{\phi})\mathbf{q}_{j}\geq 0$. We also allow for ranking restrictions on impulse responses (e.g. Amir-Ahmadi and Drautzburg 2021\nocite{Amir-Ahmadi_Drautzburg_2021}). An example is that $\eta_{ijh}(\bm{\phi},\mathbf{Q}) \geq \eta_{ijl}(\bm{\phi},\mathbf{Q})$ for $l\neq h$, which is equivalent to $(\mathbf{c}_{ih}'(\bm{\phi})-\mathbf{c}_{il}'(\bm{\phi}))\mathbf{q}_{j} \geq 0$.
    
	\item \textbf{Structural coefficients.} A restriction on the matrix of contemporaneous structural coefficients is $\mathbf{e}_{j,n}'\mathbf{A}_{0}\mathbf{e}_{i,n} \geq 0$, which is equivalent to $(\bm{\Sigma}_{tr}^{-1}\mathbf{e}_{i,n})'\mathbf{q}_{j} \geq 0$. We can also consider restrictions on $\mathbf{A}_{+} = \mathbf{Q}'\bm{\Sigma}_{tr}^{-1}\mathbf{B}$.
    
	\item \textbf{Bounds on (relative) impulse responses and elasticities.} \citet{Kilian_Murphy_2012} impose bounds on ratios of impulse responses; for example, $\eta_{i10}(\bm{\phi},\mathbf{Q})/\eta_{110}(\bm{\phi},\mathbf{Q}) \geq \lambda$, where $\lambda$ is a known scalar. We can similarly allow for bounds on ratios of elements of $\mathbf{A}_{0}$, which can be interpreted as structural elasticities (Baumeister and Hamilton 2024\nocite{Baumeister_Hamilton_2024}).
    
	\item \textbf{Narrative restrictions.} Narrative restrictions are restrictions on structural shocks in specific episodes (\citealt{Antolin-Diaz_Rubio-Ramirez_2018,Ludvigson_Ma_Ng_2018,Ludvigson_Ma_Ng_2021}). For example, the restriction that shock $j$ in period $k$ was nonnegative is $\varepsilon_{jk} = (\bm{\Sigma}_{tr}^{-1}\mathbf{u}_{k})'\mathbf{q}_{j} \geq 0$. The restriction on the historical decomposition that shock $j$ was the `most important contributor' to the observed unexpected change in variable $i$ between periods $k$ and $k+h$ is $|H_{i,j,k,k+h}|\geq \max_{l\neq j}|H_{i,l,k,k+h}|$, where $H_{i,j,k,k+h} = \sum_{l=0}^{h}\mathbf{c}_{il}'(\bm{\phi})\mathbf{q}_{j}\mathbf{q}_{j}'\bm{\Sigma}_{tr}^{-1}\mathbf{u}_{k+h-l}$.
    
	\item \textbf{Other restrictions.} We allow for other types of inequality restrictions, including on forecast error variance decompositions (e.g. \citealt{Volipcella_2021}), and relationships between proxy variables and structural shocks (e.g. \citealt{Ludvigson_Ma_Ng_2018,Ludvigson_Ma_Ng_2021,Arias_Rubio-Ramirez_Waggoner_2021,Giacomini_Kitagawa_Read_2022b,Braun_Bruggemann_2022}).\footnote{We do not allow for exogeneity restrictions related to proxy variables, which are equivalent to zero restrictions. Examples of inequality restrictions involving proxy variables include the restriction that a proxy is positively correlated with a shock or that the contribution of a shock to the variance of a proxy is greater than the contributions of other shocks (e.g. \citealt{Braun_Bruggemann_2022}).} These restrictions can all be cast as inequality restrictions on $\mathbf{Q}$.
\end{itemize}

Let $S(\bm{\phi},\mathbf{Q}) \geq \mathbf{0}_{s\times 1}$ represent a generic collection of $s$ sign restrictions.\footnote{Some restrictions depend on parameters or objects not captured in the definition of $\bm{\phi}$. For instance, narrative restrictions depend on the reduced-form innovations in specific periods. We leave this potential dependence implicit.} Given the sign restrictions, the identified set for $\mathbf{Q}$ is
\begin{equation}
	\mathcal{Q}(\bm{\phi} \mid S) = \left\{\mathbf{Q} \in \mathcal{O}(n): S(\bm{\phi},\mathbf{Q})\geq\mathbf{0}_{s\times 1}\right\}.
\end{equation}
The identified set $\mathcal{Q}(\bm{\phi} \mid S)$ is the set of orthonormal matrices that are consistent with the joint distribution of the data (summarised by $\bm{\phi}$) and the imposed identifying restrictions. Given $\bm{\phi}$, the identified set consists of observationally equivalent values of $\mathbf{Q}$, which correspond to the same value of the likelihood function (\citealt{Rothenberg_1971}). $\mathcal{Q}(\bm{\phi} \mid S)$ induces identified sets for other parameters, such as impulse responses. For example, the identified set for $\eta_{ijh}(\bm{\phi},\mathbf{Q})$ is $\{\eta_{ijh}(\bm{\phi},\mathbf{Q}): \mathbf{Q} \in \mathcal{Q}(\bm{\phi}\mid S)\}$.

\subsection{Sampling problem}
\label{subsec:bayesianinference}

We focus on the problem of sampling $\mathbf{Q}$ from the uniform distribution (Haar measure) over $\mathcal{Q}(\bm{\phi} \mid S)$:
\begin{equation}\label{eq:samplingproblem}
    \pi_{\mathbf{Q}|\bm{\phi}}(\mathbf{Q}\mid \bm{\phi}; S) = \frac{\mathbbm{1}\left(S(\bm{\phi},\mathbf{Q}) \geq \mathbf{0}_{s\times 1}\right)}{\int_{\tilde{Q} \in \mathcal{Q}(\bm{\phi}\mid S)} \mathbbm{1}(S(\bm{\phi},\tilde{\mathbf{Q}}) \geq \mathbf{0}_{s\times 1})d\tilde{\mathbf{Q}}}.
\end{equation}
Sampling from this distribution is required when conducting Bayesian inference under a conditionally uniform prior for $\mathbf{Q}$ (\citealt{DelNegro_Schorfheide_2011,Uhlig_2017,Amir-Ahmadi_Drautzburg_2021}). Under this prior, posterior draws can be obtained by drawing $\bm{\phi}$ from its posterior (typically normal-inverse-Wishart) and obtaining a fixed number of draws of $\mathbf{Q}$ from the uniform distribution over $\mathcal{Q}(\bm{\phi}\mid S )$. Draws of $\mathbf{Q}$ from within $\mathcal{Q}(\bm{\phi}\mid S )$ can also be used to approximate identified sets, providing a key input for prior-robust Bayesian inference procedures (\citealt{Giacomini_Kitagawa_2021,Giacomini_Kitagawa_Read_2022b,Giacomini_Kitagawa_Read_2023,Giacomini_Kitagawa_Read_2026}) and some frequentist methods.\footnote{Computing bounds of identified sets is required when constructing frequentist estimators of identified sets (e.g. \citealt{Ludvigson_Ma_Ng_2021}) and when implementing some frequentist inference procedures in set-identified SVARs (e.g. \citealt{Gafarov_Meier_MontielOlea_2025}).}

\cite{Arias_Rubio-Ramirez_Shin_2026} develop a Gibbs sampler to sample from the posterior of the orthogonal reduced-form parameters given an unconditionally uniform prior for $\mathbf{Q}$, which is the prior considered in \cite{Arias_Rubio-Ramirez_Waggoner_2018,Arias_Rubio-Ramirez_Waggoner_2025}. Their sampler requires sampling $\mathbf{Q}$ from a conditional density that is proportional to (\ref{eq:samplingproblem}). The approach that we propose could therefore also be used to sample from a posterior induced by the unconditionally uniform prior for $\mathbf{Q}$.

\section{Algorithms}
\label{sec:algorithms}
This section describes algorithms that can be used to draw $\mathbf{Q}$ from the uniform distribution over $\mathcal{Q}(\bm{\phi}\mid S )$. The algorithms can be combined with a sampler for $\bm{\phi}$ to draw from a joint posterior (or prior) distribution for the orthogonal reduced-form parameters. As a benchmark, we first describe an accept-reject algorithm. We then introduce our general approach to sampling based on soft sign restrictions, before describing a specific MCMC sampler -- the slice sampler -- that we use to implement our general approach.

\subsection{Accept-reject sampling}
\label{subsec:acceptreject}

The following algorithm describes an accept-reject sampler for drawing from the uniform distribution over $\mathcal{Q}(\bm{\phi}\mid S)$.
\begin{algorithm}
	[\textbf{Accept-reject sampling}]\label{algo:acceptreject}
	For a given value of $\bm{\phi}$:
	\begin{enumerate}[label=\arabic*)]
		\item Draw an $n\times n$ matrix $\mathbf{Z}$ of independent standard normal random variables and let $\mathbf{Z} = \mathbf{Q}\mathbf{R}$ be the QR decomposition of $\mathbf{Z}$, where $\mathbf{Q}$ is orthonormal and $\mathbf{R}$ is upper-triangular with nonnegative diagonal elements.
		\item Keep the draw if it satisfies $S(\bm{\phi},\mathbf{Q}) \geq \mathbf{0}_{s\times 1}$ and terminate the algorithm. Otherwise, return to Step~1.
	\end{enumerate}
\end{algorithm}
Step~1 draws $\mathbf{Q}$ from a uniform distribution (or Haar measure) over $\mathcal{O}(n)$ (\citealt{Stewart_1980,Rubio-Ramirez_Waggoner_Zha_2010}). Step~2 is the accept-reject step. The algorithm is repeated to obtain the desired number of draws.

Let $Q(\mathbf{Z})$ denote the orthonormal matrix $\mathbf{Q}$ in the QR decomposition of $\mathbf{Z}$, defined by $\mathbf{Z} = \mathbf{Q}\mathbf{R}$. Algorithm~\ref{algo:acceptreject} can be interpreted as drawing $\mathbf{Z}$ from a truncated standard matrix normal distribution with density
\begin{equation}\label{eq:truncateddist}
	f(\mathbf{Z} \mid Q(\mathbf{Z}) \in \mathcal{Q}(\bm{\phi}\mid S )) = \frac{f_{Z}(\mathbf{Z})\mathbbm{1}(Q(\mathbf{Z}) \in \mathcal{Q}(\bm{\phi}\mid S ))}{\int_{\{\tilde{\mathbf{Z}}: Q(\tilde{\mathbf{Z}}) \in \mathcal{Q}(\bm{\phi}\mid S )\}}f_{Z}(\tilde{\mathbf{Z}})d\tilde{\mathbf{Z}}},
\end{equation}
where $f_{Z}(\mathbf{Z})$ is the standard matrix normal density.\footnote{For a definition of the standard matrix normal distribution, see \cite{Gupta_Nagar_2000}.} The density in (\ref{eq:truncateddist}) is well-defined provided that the probability
\begin{equation}
    p(\bm{\phi}\mid S )\equiv \int_{\{\mathbf{Z}: Q(\mathbf{Z}) \in \mathcal{Q}(\bm{\phi}\mid S )\}}f_{Z}(\mathbf{Z})d\mathbf{Z}
\end{equation}
is strictly positive. When embedding the accept-reject sampler within a posterior sampler for $\bm{\phi}$, we assume that $p(\bm\phi\mid S )\geq \underline p>0$ for almost every $\bm{\phi}$. This rules out cases where $\mathcal{Q}(\bm{\phi} \mid S)$ is empty (which we discuss further below) or where $\mathcal{Q}(\bm{\phi} \mid S)$ lies on a lower-dimensional submanifold of $\mathcal{O}(n)$; for example, it excludes inequality restrictions that together imply an equality (e.g. $\eta_{ijh} \geq 0$ and $\eta_{ijh} \leq 0$).

The challenge with using accept-reject sampling in this setting is that it may take a large number of candidate draws (and thus computational time) to obtain a sufficiently large number of draws satisfying the identifying restrictions. This will occur when identification is `tight' and $p(\bm\phi\mid S )$ is small.

\subsection{Soft sign restrictions}
\label{subsec:softsign}

The indicator function $\mathbbm{1}(Q(\mathbf{Z}) \in \mathcal{Q}(\bm{\phi}\mid S ))$ in (\ref{eq:truncateddist}) can be decomposed into a product of indicator functions corresponding to individual sign restrictions:
\begin{equation}
	\mathbbm{1}(Q(\mathbf{Z}) \in \mathcal{Q}(\bm{\phi}\mid S )) = \prod_{l=1}^{s}\mathbbm{1}(Q(\mathbf{Z}) \in \mathcal{Q}(\bm{\phi}\mid S ^{(l)})),
\end{equation}
where $\mathcal{Q}(\bm{\phi}\mid S ^{(l)}) = \left\{\mathbf{Q} \in \mathcal{O}(n): S^{(l)}(\bm{\phi},\mathbf{Q}) \geq 0\right\}$ and $S^{(l)}(\bm{\phi},\mathbf{Q}) \geq 0$ represents the $l$th sign restriction with $l=1,\ldots,s$. The key feature underlying our approach is that we replace these indicator functions with a smooth regularisation function $\Lambda(x,\Delta)$ satisfying the following assumption.

\begin{assumption} \label{ass1}
	The regularisation function $\Lambda(x,\Delta) : \mathbb{R} \times \mathbb{R}_{+} \rightarrow (0,1)$ satisfies
	\begin{align*}  
		\lim_{x \rightarrow \infty} \Lambda(x,\Delta) &= 1 \quad \forall \Delta \in \mathbb{R}_{+} \\
		\lim_{x\rightarrow -\infty} \Lambda(x,\Delta) &= 0 \quad \forall \Delta \in \mathbb{R}_{+} \\
		\lim_{\Delta \rightarrow 0} \Lambda(x,\Delta) &= 
		\begin{cases}
			1 & x \geq 0 \\
			0 & x < 0
		\end{cases}
	\end{align*}
	In addition, for some finite $K>0$, $\Lambda(x,\Delta)$ satisfies
	\[
	|\Lambda(x,\Delta)-\mathbbm{1}(x\geq0)|\leq K
	\]
	for all $\Delta\in\mathbb{R}_+$ and $x\in\mathbb{R}$.
\end{assumption}
$\Lambda(x,\Delta)$ can be interpreted as smoothly penalising values of $\mathbf{Q}$ (equivalently, $\mathbf{Z}$) that violate (or are close to violating) the sign restrictions. In the limit as $\Delta \rightarrow 0$, the regularisation function converges to the indicator function.

One choice for $\Lambda(x,\Delta)$ that satisfies Assumption~\ref{ass1} (and that we will make use of below) is the logistic function:
\begin{equation}\label{eq:logistic}
	\Lambda(x,\Delta)=\frac{1}{1+\exp(-x/\Delta)}.
\end{equation}
We proceed with this choice because it is computationally inexpensive to evaluate relative to some alternatives (e.g. the standard normal cumulative distribution function (CDF)). The specific choice of regularisation function has no substantive bearing on the properties of our approach (e.g. the quality of posterior approximations), though it may affect computational speed (e.g. if the chosen regularisation function is expensive to evaluate).

We propose sampling from a smooth density that replaces the indicator function with the regularisation function:
\begin{equation}
	f_{\Delta}(\mathbf{Z}) = \frac{f_{Z}(\mathbf{Z})\prod_{l=1}^{s}\Lambda(S^{(l)}(\bm{\phi},Q(\mathbf{Z})),\Delta)}{\int f_{Z}(\tilde{\mathbf{Z}})\prod_{l=1}^{s}\Lambda(S^{(l)}(\bm{\phi},Q(\tilde{\mathbf{Z}})),\Delta)d\tilde{\mathbf{Z}}}.
\end{equation}
Replacing the indicator function with the regularisation function smoothly penalises values of $\mathbf{Z}$ (or $\mathbf{Q}$) that violate the sign restrictions by downweighting their density. The advantage of working with this smooth density is that alternative sampling algorithms, such as MCMC methods, can be directly applied, obviating the need for accept-reject sampling. In the limit as $\Delta \rightarrow 0$, $f_{\Delta}(\mathbf{Z})$ approaches the truncated density $f(\mathbf{Z} \mid Q(\mathbf{Z}) \in \mathcal{Q}(\bm{\phi}\mid S ))$. This claim is formalised in the following proposition.

\begin{proposition}\label{prop:conv}
	Assume the conditions in Assumption~\ref{ass1} hold and let $T:\mathbb{R}^d\rightarrow \mathbb{R}$ be such that $\int_{\mathbb{R}^d}|T(\mathbf{Z})|f_Z(\mathbf{Z})d\mathbf{Z}<\infty$. Then,
	\begin{equation}
		\lim_{\Delta\rightarrow 0}\left|\mathbb{E}_f(T(\mathbf{Z}))-\mathbb{E}_\Delta(T(\mathbf{Z}))\right| = 0,
	\end{equation}
	where $\mathbb{E}_f(.)$ and $\mathbb{E}_\Delta(.)$ are expectations taken under $f$ and $f_\Delta$, respectively.
\end{proposition}
For $\Delta > 0$, draws of $\mathbf{Z}$ obtained from $f_{\Delta}(\mathbf{Z})$ will not necessarily satisfy the sign restrictions and -- conditional on satisfying the sign restrictions -- will not follow the desired truncated normal distribution; equivalently, draws of $\mathbf{Q}$ will not be uniformly distributed over $\mathcal{Q}(\bm{\phi}\mid S)$. However, an importance-sampling step can be applied to obtain draws from the uniform distribution. The importance weights are given by
\begin{equation}\label{eq:iweights}
	\frac{f(\mathbf{Z} \mid  Q(\mathbf{Z}) \in \mathcal{Q}(\bm{\phi}\mid S))}{f_{\Delta}(\mathbf{Z})} \propto \frac{\mathbbm{1}(Q(\mathbf{Z}) \in \mathcal{Q}(\bm{\phi}\mid S))}{\prod_{l=1}^{s}\Lambda(S^{(l)}(\bm{\phi},Q(\mathbf{Z})),\Delta)}.
\end{equation}
The normalising constant for the importance weights is computationally costly to obtain, so we use a self-normalised importance sampler, where the importance weights on the right-hand side of (\ref{eq:iweights}) are normalised to sum to one across draws from $f_{\Delta}(\mathbf{Z})$. Self-normalisation induces a small finite-sample bias but yields consistent estimators as the number of importance draws increases (e.g. \citealt{Robert_Casella_2004}). A corollary of Proposition~\ref{prop:conv} is that the normalising constant converges to one as $\Delta \rightarrow 0$. This implies that any finite-sample bias should be small for small $\Delta$.

A small $\Delta$ can introduce sampling inefficiencies as $f_{\Delta}(\mathbf{Z})$ will have regions where the log target density changes sharply. In a random walk Metropolis algorithm, such features require a relatively small tuning parameter (i.e. proposal scale) to achieve reasonable acceptance rates, since larger steps are more likely to be rejected in regions with high local curvature. In the next section, we discuss an alternative method -- slice sampling -- that is more robust in such settings and can offer improved efficiency when exploring irregular target distributions.

For importance sampling to produce valid inference, the importance weights must have finite variance (\citealt{Geweke_1989,Koopman_Shephard_Creal_2009}). Assuming that $\Lambda(x,\Delta)$ is logistic, the following proposition shows that the (unnormalised) importance weights are uniformly bounded for all $\Delta > 0$, independently of the omitted normalising constant. This immediately implies that their variance is finite.\footnote{The same bounds apply under any regularisation function satisfying Assumption~\ref{ass1} and where $\Lambda(x,\Delta) \geq 1/2$ for all $x\geq 0$; this would be the case when $\Lambda(x,\Delta)$ is the CDF of a random variable with median zero (e.g. the standard normal CDF).}

\begin{proposition}\label{prop:bounded_weights}
Let the (unnormalised) importance weight be
\[
\tilde w_{\Delta}(\mathbf{Z})
\;\equiv\;
\frac{\mathbbm{1}\!\left(Q(\mathbf{Z})\in \mathcal Q(\bm{\phi}\mid S)\right)}
{\prod_{\ell=1}^s \Lambda\!\left(S^{(\ell)}(\bm{\phi},Q(\mathbf{Z})),\Delta\right)},
\]
where $\Lambda(x,\Delta)$ is the logistic function in (\ref{eq:logistic}). $\tilde w_{\Delta}(\mathbf{Z})$ satisfies
\[
0\ \le\ \tilde w_{\Delta}(\mathbf{Z})\ \le\ 2^s.
\]
\end{proposition}

The choice of $\Delta$ controls how closely the smoothed density $f_\Delta(\mathbf{Z})$ approximates the hard-truncated target density. As $\Delta \to 0$, the logistic regularisation function converges pointwise to the indicator function defining the identifying restrictions, so $f_\Delta(\mathbf{Z})$ converges to the target density in Proposition 1. In this sense, smaller values of $\Delta$ reduce the discrepancy between the proposal distribution generated by the slice sampler and the desired hard-truncated distribution. Moreover, for draws satisfying the restrictions and lying away from the boundary of the identified set, the logistic factors are close to one when $\Delta$ is small, so the importance weights are close to constant. This suggests that, other things equal, smaller values of $\Delta$ should improve the efficiency of the importance-sampling correction.

There are nevertheless two practical considerations. First, when the identified set is disconnected, larger values of $\Delta$ may facilitate movement between disconnected components by assigning non-negligible probability to regions outside the identified set. Very small values of $\Delta$ remove this soft bridge and can therefore make global mixing across disconnected components more difficult, although this is distinct from the local stickiness associated with random-walk Metropolis algorithms. Second, very small values of $\Delta$ may create numerical issues when evaluating the logistic factors, so the log target density should be evaluated using numerically stable formulae. In practice, we recommend choosing $\Delta$ small enough that posterior summaries are stable and the effective sample size of the importance weights is high, while also checking that the Markov chain explores the relevant regions of the identified set. In our empirical application, posterior summaries
are stable over a broad range of small values of $\Delta$, and the effective number of draws per unit of computation is highest for values around $10^{-3}$ to $10^{-4}$.

\subsection{Slice sampling}
\label{subsec:slice}
There are many MCMC methods that could be used to sample from $f_{\Delta}(\mathbf{Z})$. We use the slice sampler, motivated by its robust convergence properties, efficiency (relative to standard random walk Metropolis) and ease of implementation (\citealt{Neal_2003}).

The slice sampler is motivated by the fact that sampling from $f_{\Delta}(\mathbf{Z})$ is equivalent to sampling uniformly from the region under the density function. The `simple' slice sampler constructs a Markov chain that converges to this uniform distribution by alternating between two steps: 1) sample $y$ uniformly from $[0,f_{\Delta}(\mathbf{Z}_{k})]$ given some predetermined $\mathbf{Z}_{k}$; and 2) sample $\mathbf{Z}_{k+1}$ uniformly from the `slice' $\mathcal{S}(y) = \{\mathbf{Z}: f_{\Delta}(\mathbf{Z}) > y\}$.\footnote{It is only necessary to evaluate $f_{\Delta}(\mathbf{Z})$ up to a constant of proportionality.} Iterating over these steps generates dependent draws from the target density.

Theoretical results for the simple slice sampler provide motivation for using slice-based methods in this setting. For example, \citet{Mira_Tierney_2002} prove that, if the target density is bounded and has support with finite Lebesgue measure, the simple slice sampler is uniformly ergodic. Moreover, as noted by \citet{Roberts_Rosenthal_1999}, the simple slice sampler is almost always geometrically ergodic. These results are useful because they highlight the robustness of slice sampling relative to many random-walk methods, particularly for non-standard target densities. However, the simple slice sampler requires drawing exactly from the full slice $\mathcal{S}(y)$, which is infeasible in the multivariate setting considered here. Our implementation therefore follows the multivariate shrinking-hyperrectangle procedure described by \citet{Neal_2003}, with a contaminated choice of hyperrectangle width to encourage both local exploration and occasional larger moves. To summarise, the procedure: randomly positions a hyperrectangle with side length $w$ around the initial point; draws a point uniformly from the hyperrectangle; and repeatedly shrinks the hyperrectangle (`shrinking in') if the candidate lies outside the slice until an admissible draw is obtained.

The uniform- and geometric-ergodicity results for the simple slice sampler do not directly imply identical guarantees for this implemented multivariate sampler. The relevant property of the shrinking procedure is instead invariance: conditional on the auxiliary slice height, the hyperrectangle-shrinking step is constructed to leave the uniform distribution over the slice invariant. Consequently, the resulting Markov chain has $f_{\Delta}(\mathbf{Z})$ as its invariant density under the usual irreducibility and aperiodicity conditions for the implemented transition kernel. We therefore use the ergodicity results for the simple slice sampler as motivation for adopting a slice-sampling approach, while relying on invariance of the shrinking procedure to justify the target distribution of the implemented algorithm.

The choice of side length $w$ will affect the sampler's convergence speed and computational efficiency. Under the general class of identifying restrictions we consider, there is no guarantee that $\mathcal{Q}(\bm{\phi}\mid S)$ is path connected, implying that $f_{\Delta}(\mathbf{Z})$ may be multimodal. A small value of $w$ can therefore lead to poor mixing across modes and slow convergence. Conversely, setting $w$ too large may reduce computational efficiency, as many shrinking-in steps can be required to obtain an admissible draw. We balance these considerations by using a contaminated proposal, where $w=2$ with 95~per cent probability and $w=6$ with 5~per cent probability. Choices of $w$ on this scale are reasonable given that, for small $\Delta$, the target distribution resembles a truncated standard matrix normal distribution.

We are now in a position to describe our algorithm for sampling $\mathbf{Q}$ from the uniform distribution over $\mathcal{Q}(\bm{\phi} \mid S)$, implemented using soft sign restrictions and the slice sampler.
\begin{algorithm}
    [\textbf{Sampling via soft sign restrictions and the slice sampler}.]\label{algo:softsign}
    Given $\bm{\phi}$, choices for the regularisation function $\Lambda(x,\Delta)$ and regularisation parameter $\Delta$, and initial value $\mathbf{Z}^{(0)}$:
    \begin{enumerate}[label=\arabic*)]
        \item Obtain $M$ draws $\mathbf{z}^{(m)} = \mathrm{vec}(\mathbf{Z}^{(m)})$ from $f_{\Delta}(\mathbf{Z})$ using the slice sampler with initial value $\mathbf{z}^{(0)}$. For $m=1,\ldots,M$:
        \begin{enumerate}[label=\roman*)]
            \item Sample $y \sim \mathrm{Uniform}(0,f_{\Delta}(\mathbf{Z}^{(m-1)}))$.
            \item Randomly position a hyperrectangle $\mathcal{H}$ with side length $w$ around $\mathbf{z}^{(m-1)}$, where $w = 2$ with probability $\alpha$ and $w=6$ with probability $1-\alpha$:
            \begin{equation*}
                \mathcal{H} = (L_1,R_1)\times \ldots \times (L_{n^2},R_{n^2}),
            \end{equation*}
            where $L_i = z_{i}^{(m-1)} - wU_i$, $R_i = L_i + w$ and $U_i \sim \mathrm{Uniform}(0,1)$ for $i=1,\ldots,n^2$.
            \item Draw $\tilde{\mathbf{z}} = (\tilde{z}_1,\ldots,\tilde{z}_{n^2})'$ from a uniform distribution over $\mathcal{H}$: $\tilde{z}_i = L_i + U_i(R_i - L_i)$ for $i=1,\ldots,n^2$.
            \item If $f_{\Delta}(\tilde{\mathbf{Z}}) > y$ set $\mathbf{Z}^{(m)} = \tilde{\mathbf{Z}}$ and return to (i). Otherwise, shrink $\mathcal{H}$ by setting $L_i = \tilde{z}_i$ if $\tilde{z}_i < z_{i}^{(m-1)}$ or $R_{i} = \tilde{z}_i$ if $\tilde{z}_i > z_{i}^{(m-1)}$ (for $i=1,\ldots,n^2$), and return to (iii). 
        \end{enumerate}
        \item For $m=1,\ldots,M$, compute $\mathbf{Q}^{(m)} = Q(\mathbf{Z}^{(m)})$, where $Q(\mathbf{Z}^{(m)})$ is the orthonormal matrix in the QR decomposition of $\mathbf{Z}^{(m)}$, and evaluate the (unnormalised) importance weights $\tilde w_{\Delta}(\mathbf{Z}^{(m)})$.
        \item Obtain $K$ resampled draws using $\tilde w_{\Delta}(\mathbf{Z})$ as importance weights.
    \end{enumerate}
\end{algorithm}
In practice we select the initial value $\mathbf{Z}^{(0)}$ using numerical optimisation to find a (potentially local) maximum of $\log f_{\Delta_0}(\mathbf{Z})$ with $\Delta_0 \geq \Delta$; experiments suggest that this initialisation strategy increases the efficiency of the sampler relative to initialising at a random draw.\footnote{We use the `fminsearch' algorithm in MATLAB, which is a Nelder-Mead simplex algorithm (\citealt{Lagarias_etal_1998}). The optimiser is initialised at a random draw of $\mathbf{Z}$ from $f_{Z}(\mathbf{Z})$.} Step~1 uses the slice sampler described in Section~\ref{subsec:slice} to obtain $M$ draws of $\mathbf{Z}$ from the smoothed target density $f_{\Delta}(\mathbf{Z})$. As we have noted above, other MCMC algorithms could be applied in this step. For each draw of $\mathbf{Z}$, Step~2 computes the associated draw of $\mathbf{Q}$ and its importance weight. Finally, Step~3 resamples the draws using importance sampling; because $w_{\Delta}(\mathbf{Z}^{(k)})=0$ whenever $Q(\mathbf{Z}^{(k)}) \notin \mathcal{Q}(\bm{\phi} \mid S)$, the importance-sampling step discards draws that violate the identifying restrictions. In the examples and applications below, we set $K=M$. The sampler can be embedded within a sampler for $\bm{\phi}$.

If the draws of $\mathbf{Q}$ are used to approximate the bounds of an identified set, such as when conducting prior-robust Bayesian inference, resampling the draws in Step~4 is unnecessary and it suffices to discard draws that violate the sign restrictions (i.e. with $\tilde w_{\Delta}(\mathbf{Z}) = 0$). This is because the approximated bounds depend only on the extrema of the parameter of interest over the identified set and not on its distribution.

\subsection{Advantages of `softening'}
\label{sec:softening}

The accept-reject algorithm samples from a truncated density that assigns zero probability outside the identified set. Consequently, the algorithm discards proposal draws without regard to their distance from the boundary of the identified set. For example, draws that narrowly violate a single sign restriction are treated identically to draws that violate many restrictions by a wide margin, so the sampler does not learn about the geometry or location of the identified set.

Smoothly penalising parameter values that violate the identifying restrictions converts the truncated density into a continuous target density that is well suited to MCMC sampling. MCMC methods naturally leverage local information to move toward and within high-density regions. Under the smoothed density, parameter values that are closer to satisfying the restrictions receive higher probability mass than those that are further away, creating a gradient that guides the Markov chain toward the identified set. This allows MCMC algorithms to exploit local information to move efficiently toward the identified set and to explore it without repeatedly proposing infeasible draws. In contrast to accept–reject sampling, near-feasible draws are no longer discarded but instead play an active role in directing the sampler.

MCMC methods could in principle be applied to directly sample from the truncated distribution $f(\mathbf{Z} \mid  Q(\mathbf{Z}) \in \mathcal{Q}(\bm{\phi}\mid S))$. Softening the sign restrictions has two main advantages over such an approach. 

First, an important practical consideration is that the slice sampler, like other related MCMC samplers, requires an initial value $\mathbf{Z}_0$ that lies within the identified set, so that it is assigned positive density under the truncated distribution. In problems where $\mathcal{Q}(\bm{\phi}\mid S)$ is small relative to $\mathcal{O}(n)$, constructing such an initial value may be nontrivial; na{\"i}ve random initialisation may fail with very high probability and deterministic feasibility searches can be costly. Consequently, the overall runtime of the sampler may be dominated by the effort required to obtain a feasible initial value. Our approach eliminates the requirement that the chain be initialised within the identified set. Since $f_{\Delta}(\mathbf{Z}) > 0$ for all $\mathbf{Z}$, the sampler can start from a generic point and drift toward regions where the identifying restrictions are satisfied, rather than relying on a costly search for a feasible initial value.

Second, our approach can improve global mixing when the identified set is not path-connected, in which case the truncated density features multiple modes separated by regions of zero probability density. Under hard constraints, the chain can transition between disconnected regions only if the proposal intersects another region of the identified set, which may occur with very low probability. Under the smooth density $f_{\Delta}(\mathbf{Z})$, regions that lie outside the identified set receive small but nonzero probability, providing a continuous `bridge' through which the chain may move between modes. This effect is analogous to tempering: larger $\Delta$ makes $f_{\Delta}(\mathbf{Z})$ flatter and easier to explore, while smaller $\Delta$ makes it concentrate near the identified set.

In summary, the smoothed density $f_{\Delta}(\mathbf{Z})$ is a computationally convenient approximation to the hard-truncated distribution that facilitates initialisation and exploration of the parameter space. Moreover, because $f_{\Delta}(\mathbf{Z})$ coincides with the truncated distribution in the limit as $\Delta \rightarrow 0$, small values of $\Delta$ yield a proposal distribution that is suitable for importance sampling.

\subsection{Empty identified sets}
\label{sec:emptyset}

The discussion above has assumed that $\mathcal{Q}(\bm{\phi}\mid S)$ is nonempty at the value of $\bm{\phi}$ under consideration (e.g. a posterior draw or the maximum likelihood estimator). However, $\mathcal{Q}(\bm{\phi}\mid S)$ may be empty at some (or all) values of $\bm{\phi}$. An empty $\mathcal{Q}(\bm{\phi}\mid S)$ means that the identifying restrictions are incompatible with $\bm{\phi}$. Ideally, we would know whether $\mathcal{Q}(\bm{\phi}\mid S)$ is nonempty before attempting to sample from it. While there are algorithms to verify whether identified sets are nonempty, these are applicable only under specific types of identifying restrictions.\footnote{\cite{Amir-Ahmadi_Drautzburg_2021}, \cite{Giacomini_Kitagawa_Volpicella_2022} and \cite{Read_2022} propose algorithms for verifying whether identified sets are nonempty when there are linear restrictions on a single column of $\mathbf{Q}$.} 

Following \cite{Giacomini_Kitagawa_2021}, we treat $\mathcal{Q}(\bm{\phi}\mid S)$ as empty when, after $M$ attempted draws of $\mathbf{Q}$, none satisfy the identifying restrictions. With finite $M$, a nonempty $\mathcal{Q}(\bm{\phi}\mid S)$ may be misclassified as empty, although the probability of misclassification vanishes as $M \rightarrow \infty$. This issue applies to both accept-reject sampling and our sampler.\footnote{Both samplers have zero probability of classifying $\mathcal{Q}(\bm{\phi}\mid S)$ as nonempty when it is empty.} In the empirical applications below, our sampler has a higher probability of correctly classifying the identified set as nonempty than accept-reject (for the same $M$). This is because our sampler tends to drift towards $\mathcal{Q}(\bm{\phi}\mid S)$, as discussed above. More generally, when embedding our sampler within a sampler for $\bm{\phi}$, repeated failure to obtain draws of $\bm{\phi}$ such that $\mathcal{Q}(\bm{\phi}\mid S)$ is nonempty suggests that the identifying restrictions need to be relaxed.\footnote{Emptiness of $\mathcal{Q}(\bm{\phi}\mid S)$ at posterior draws of $\bm{\phi}$ does not necessarily imply that the identifying restrictions are incorrect. To establish this, one would need to verify whether $\mathcal{Q}(\bm{\phi}_0\mid S)$ is empty, where $\bm{\phi}_0$ is the (pseudo) true value of $\bm{\phi}$. However, when $\mathcal{Q}(\bm{\phi}\mid S)$ is empty with high posterior probability, this can be interpreted as evidence against the `plausibility' of the restrictions; see \cite{Giacomini_Kitagawa_2021} and \cite{Giacomini_Kitagawa_Read_2022a} for further discussion.} 

\section{Numerical Illustrations}
\label{sec:numericalexamples}
This section uses a bivariate model to illustrate our method and compare it with accept-reject sampling. The model's simplicity allows us to transparently control the size of the identified set and facilitates visualisation. We first consider a case where $\mathcal{Q}(\bm{\phi} \mid S)$ is path-connected and then examine a more challenging case where $\mathcal{Q}(\bm{\phi} \mid S)$ consists of disconnected regions.

\subsection{Connected identified set}
\label{subsec:connected}

Let $\mathbf{y}_{t} = (p_{t},q_{t})'$ contain log price and quantity. Assume $\mathbf{y}_t$ follows the SVAR($0$) $\mathbf{y}_t = \mathbf{A}_0^{-1}\bm{\varepsilon}_t = \bm{\Sigma}_{tr}\mathbf{Q}\mathbf{u}_t$ and let $\bm{\phi}=\mathrm{vech}(\bm{\Sigma}_{tr}) = (\sigma_{11},\sigma_{21},\sigma_{22})'$. The space of $2\times 2$ orthonormal matrices $\mathcal{O}(2)$ can be represented as
\begin{equation}
	\mathcal{O}(2) = \left\{
	\begin{bmatrix}
		\cos\theta & -\sin\theta \\
		\sin\theta & \cos\theta 
	\end{bmatrix}
	\right\} \cup
	\left\{
	\begin{bmatrix}
		\cos\theta & \sin\theta \\
		\sin\theta & -\cos\theta 
	\end{bmatrix}
	\right\},
\end{equation}
where we leave it implicit that $\theta \in [-\pi,\pi]$ (e.g. \citealt{Baumeister_Hamilton_2015}).

Consider imposing the sign restrictions
\begin{equation}
	\mathbf{A}_{0}^{-1} = 
	\begin{bmatrix}
		+ & + \\
		- & +
	\end{bmatrix},
\end{equation}
implying that the first structural shock can be interpreted as a supply shock and the second as a demand shock. Augment the sign restrictions with an upper bound on the price elasticity of supply: $\omega(\bm{\phi},\mathbf{Q}) \leq \bar{\omega}$, where
\begin{equation}
    \omega(\bm{\phi},\mathbf{Q}) \equiv -\mathbf{e}_{1,2}'\mathbf{A}_{0}\mathbf{e}_{1,2}/\mathbf{e}_{1,2}'\mathbf{A}_{0}\mathbf{e}_{2,2} = -(\bm{\Sigma}_{tr}^{-1}\mathbf{e}_{1,2})'\mathbf{q}_1/(\bm{\Sigma}_{tr}^{-1}\mathbf{e}_{2,2})'\mathbf{q}_1.
\end{equation}

Assuming that $\sigma_{21} < 0$, the identified set for $\theta$ is
\begin{equation}
	IS_{\theta}(\bm{\phi}\mid S) = \left[\arctan\left(\frac{\sigma_{22}}{\sigma_{21}}\right),\mathrm{arccot}\left(\frac{\sigma_{21}}{\sigma_{22}}-\frac{\sigma_{11}}{\sigma_{22}}\bar{\omega}\right)\right].
\end{equation}
The upper bound of the set converges to zero as $\bar{\omega} \rightarrow \infty$ and to $\arctan(\sigma_{22}/\sigma_{21})$ (i.e. the lower bound) as $\bar{\omega} \rightarrow 0$. The elasticity restriction therefore provides a convenient way to explore the efficiency of our algorithm relative to accept-reject sampling as the size of the identified set changes.

To illustrate the sampler in this setting, we fix $\bm{\phi}$ and aim to draw from the uniform distribution over $\mathcal{Q}(\bm{\phi}\mid S)$, which is equivalent to drawing $\theta$ from a uniform distribution over $IS_{\theta}(\bm{\phi}\mid S)$ (\citealt{Baumeister_Hamilton_2015}). The accept-reject algorithm can be interpreted as drawing $\theta$ from a uniform distribution over $[-\pi,\pi]$ and rejecting draws of $\theta$ that violate the sign restrictions. In contrast, our approach uses the slice sampler to draw from the distribution over $\theta$ induced by $f_{\Delta}(\mathbf{Z})$. For $\Delta >0$, this distribution assigns positive density outside $IS_{\theta}(\bm{\phi}\mid S)$, so the slice sampler will return draws of $\theta$ outside of $IS_{\theta}(\bm{\phi}\mid S)$ with positive probability, though draws of $\theta$ within $IS_{\theta}(\bm{\phi}\mid S)$ will be sampled with higher probability. The importance-sampling step discards draws outside of $IS_{\theta}(\bm{\phi}\mid S)$ and resamples the remaining draws so that they are approximately uniformly distributed.

Figure~\ref{fig:illustration} illustrates the sampler under different values of $\Delta$. When $\Delta = 100$ (top left panel), which we take to approximate the behaviour of the algorithm as $\Delta \rightarrow \infty$, values of $\theta$ that violate the sign restrictions are essentially not penalised. The slice sampler therefore generates draws of $\mathbf{Q}$ from a uniform distribution over $\mathcal{O}(2)$, which corresponds to $\theta$ being uniformly distributed over the interval $[-\pi,\pi]$. When $\Delta = 0.1$ (top right panel), values of $\theta$ that violate the sign restrictions have their density penalised, but a substantial share of draws obtained via slice sampling violate the sign restrictions. Values of $\theta$ that satisfy the sign restrictions but are close to the bounds of the identified set also have their density penalised, so the distribution of draws satisfying the sign restrictions is not uniform. Decreasing $\Delta$ (bottom two panels) more strongly penalises values of $\theta$ that violate the sign restrictions, so a smaller proportion of draws violate the restrictions and the distribution of the draws that satisfy the restrictions is closer to uniform. Following importance sampling, in all cases the draws are approximately uniformly distributed over $IS_{\theta}(\bm{\phi})$.

\begin{figure}[h]
	\centering
	\caption{Illustration of Sampling Using Soft Sign Restrictions}
	\label{fig:illustration}
	\begin{tabular}{c c}
		\includegraphics[scale=0.5]{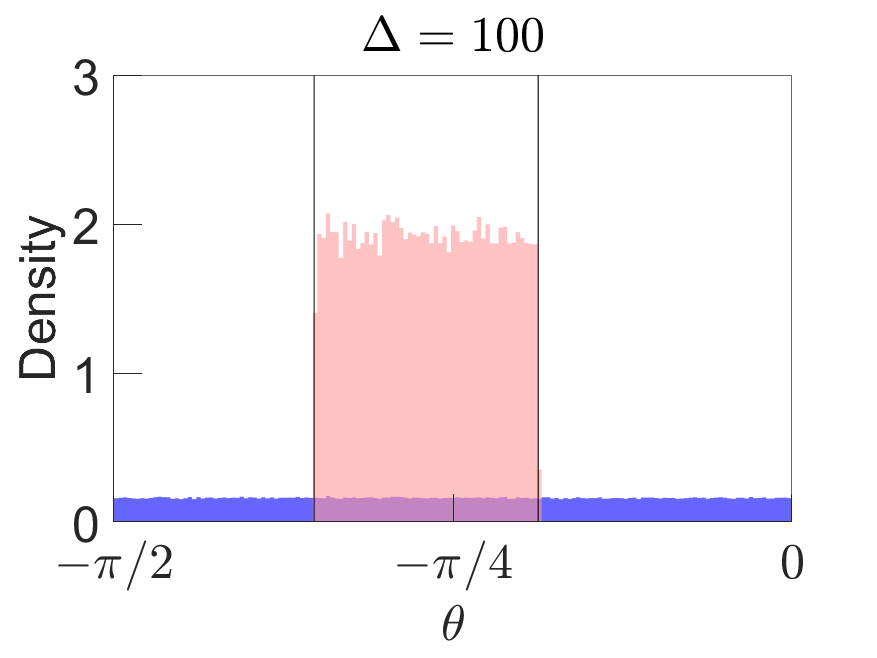}  &  \includegraphics[scale=0.5]{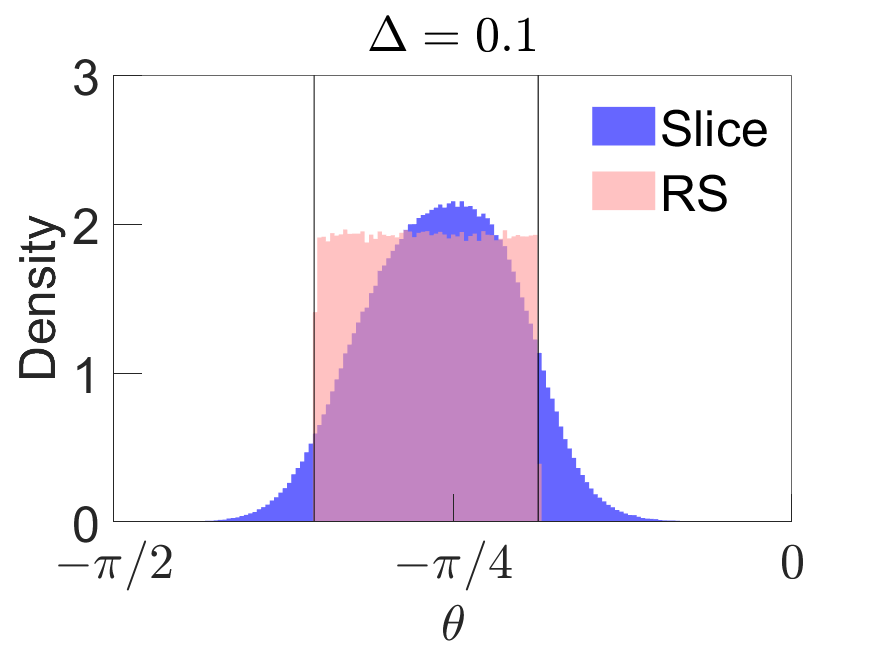} \\
		\includegraphics[scale=0.5]{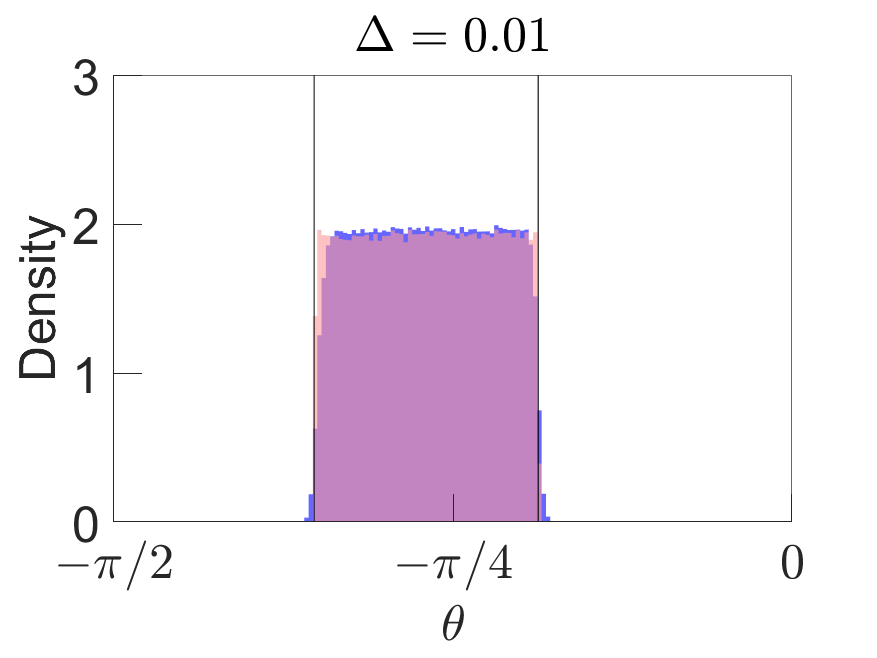}  &  \includegraphics[scale=0.5]{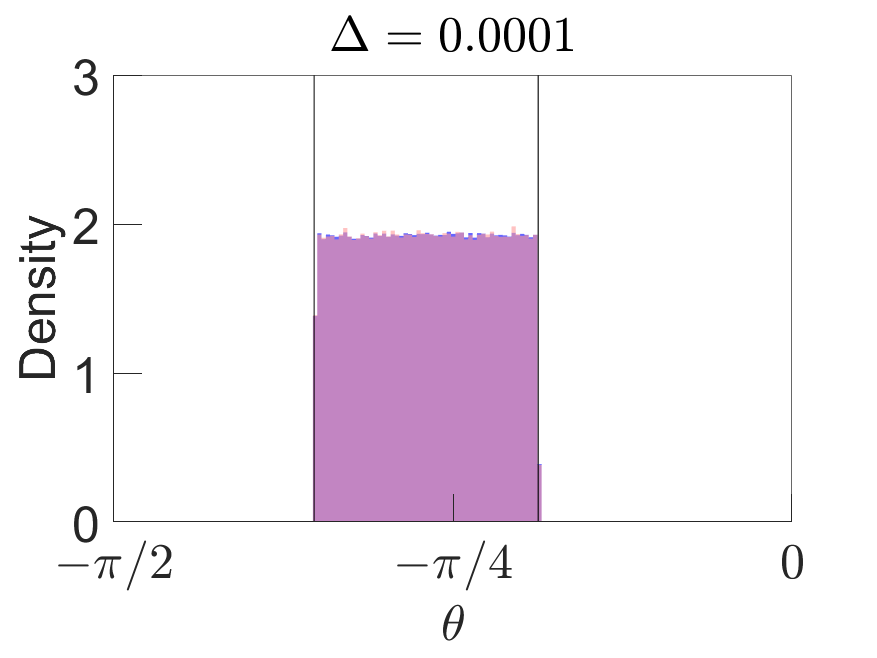} \\      
	\end{tabular}
    \noindent\begin{minipage}{\textwidth}
	\small\textbf{Notes:} `Slice' is distribution of draws obtained using slice sampler before importance sampling; `RS' is distribution after importance sampling. Vertical lines represent bounds of identified set. Parameter values are $\bm{\phi} = (\sigma_{11},\sigma_{21},\sigma_{22})' = (1,-0.5,1)'$ and $\bar{\omega} = 1$. Based on one~million draws.
     \end{minipage}   
\end{figure}

To examine the computational efficiency of the sampling algorithms in this simple setting, we obtain 10,000~draws from $\mathcal{Q}(\bm{\phi}\mid S)$ in 100~replications.\footnote{The sampler is initialised by numerically maximising the log target density with $\Delta_0 = 1000\times \Delta$, with an initial value for $\mathbf{Z}$ randomly drawn from the standard matrix normal distribution.} We compute the average time taken to obtain the draws and the average effective sample size.\footnote{If $w_{k}\equiv \tilde{w}_{\Delta}(\mathbf{Z}^{(k)})$ is the importance weight attached to the $k$th draw, the effective sample size (expressed as a percentage of the original number of draws) is\linebreak$ESS = (100/K)\times(\sum_{k=1}^{K}w_{k})^{2}/\sum_{k=1}^{K}w_{k}^{2}$.} To examine how performance varies with the tightness of identification and the choice of $\Delta$, we consider values of $\bar{\omega} \in \{1,0.1,0.01\}$ and $\Delta \in \{0.1,0.01,0.001,0.0001\}$ (Table~\ref{tab:performance2d}).\footnote{The results are obtained using MATLAB R2024a on a desktop computer running Microsoft Windows 10 Enterprise with an Intel Core i7-10700 CPU @ 2.90GHz, 8 cores and 128 GB RAM.}  

\begin{table}[h]
	\caption{Performance of Sampling Algorithms -- Bivariate Model} \label{tab:performance2d}
	\begin{tabularx}{\textwidth}{lccccccc}	\hline
		\textbf{Algorithm} & \multicolumn{3}{c}{\textbf{Speed (seconds)}} & ~~~~~~~~& \multicolumn{3}{c}{\textbf{Effective sample size (\%)}} \\ \cline{2-4} \cline{6-8}
		& $\bar{\omega} = 1$ & $\bar{\omega} = 0.1$ & $\bar{\omega} = 0.01$ & & $\bar{\omega} = 1$ & $\bar{\omega} = 0.1$ & $\bar{\omega} = 0.01$ \\ \hline
Accept-reject	&	0.77	&	5.15	&	49.32	& &	100.00	&	100.00	&	100.00	\\
Slice sampler: & & & & & & & \\
$\Delta = 0.1$	&	0.41	&	0.42	&	0.42	& &	78.36	&	22.32	&	2.17	\\
$\Delta = 0.01$	&	0.42	&	0.62	&	0.70	& &	96.52	&	80.62	&	22.72	\\
$\Delta = 0.001$	&	0.42	&	0.62	&	0.96	& &	99.65	&	97.27	&	80.92	\\
$\Delta = 0.0001$	&	0.41	&	0.62	&	0.96	& &	99.95	&	99.67	&	97.26	\\
\hline
	\end{tabularx}
    \par\vspace{0.5em}
\noindent\begin{minipage}{\textwidth}
\small\textbf{Notes:} Averages based on 100 Monte Carlo replications with 10,000 draws of $\mathbf{Q}$. $\bar{\omega}$ controls width of identified set. $\Delta$ controls penalisation of parameter values that violate (or are close to violating) sign restrictions in slice sampler.
\end{minipage}
\end{table}

In the cases considered, our sampler generates a larger number of effective draws per second. The difference in performance is less pronounced when $\bar{\omega} = 1$ and the identified set is relatively wide. As $\bar{\omega}$ decreases and the size of the identified shrinks, the relative performance of our approach improves. For example, when $\bar{\omega} = 0.01$, on average it takes the accept-reject sampler close to 50~seconds to generate 10,000 draws, whereas our sampler with $\Delta = 0.0001$ generates around 9,700 effective draws in under one~second (roughly 50~times as many effective draws per second).

\subsection{Disconnected identified set}
\label{subsec:disconnected}

In general, $\mathcal{Q}(\bm{\phi}\mid S)$ may consist of disconnected regions. Sampling from a distribution that is supported on disconnected parameter regions can pose challenges for MCMC algorithms, because the Markov chain may become stuck in one region and not adequately traverse the target distribution. In this exercise, we illustrate our sampling approach in a setting where the identified set is disconnected.

Consider imposing the restriction that $\eta_{120}(\bm{\phi},\mathbf{Q}) =  \mathbf{e}_{1,2}'\bm{\Sigma}_{tr}\mathbf{q}_{2} \geq \lambda$ for $0 \leq \lambda \leq \sigma_{11}$. All other impulse responses are unrestricted.\footnote{We continue to impose the sign normalisation $\mathrm{diag}(\mathbf{A}_0)=\mathrm{diag}(\mathbf{Q}'\bm{\Sigma}_{tr}^{-1}) \geq \mathbf{0}_{n\times 1}$. This example nests Example~B.5 in \citet{Giacomini_Kitagawa_2021sup} when $\lambda = 0$.} If $\sigma_{21} < 0$, the identified set for $\theta$ is
\begin{equation}
	\begin{split}
		\begin{array} {rcl}
	IS_{\theta}(\bm{\phi}\mid S) = \left[\arctan\left(\frac{\sigma_{22}}{\sigma_{21}}\right),\arcsin\left(-\frac{\lambda}{\sigma_{11}}\right)\right] \cup \\
	\left[\frac{\pi}{2}, \min\left\{\pi - \arcsin\left(\frac{\lambda}{\sigma_{11}}\right),\pi + \arctan\left(\frac{\sigma_{22}}{\sigma_{21}}\right)\right\} \right],
		\end{array}
	\end{split}
\end{equation}
which is the union of two disconnected intervals.

\begin{figure}[h]
	\centering
	\caption{Illustration of Sampling Using Soft Sign Restrictions -- Disconnected Identified Set}
	\label{fig:illustration_disc}
	\begin{tabular}{c c}
		\includegraphics[scale=0.5]{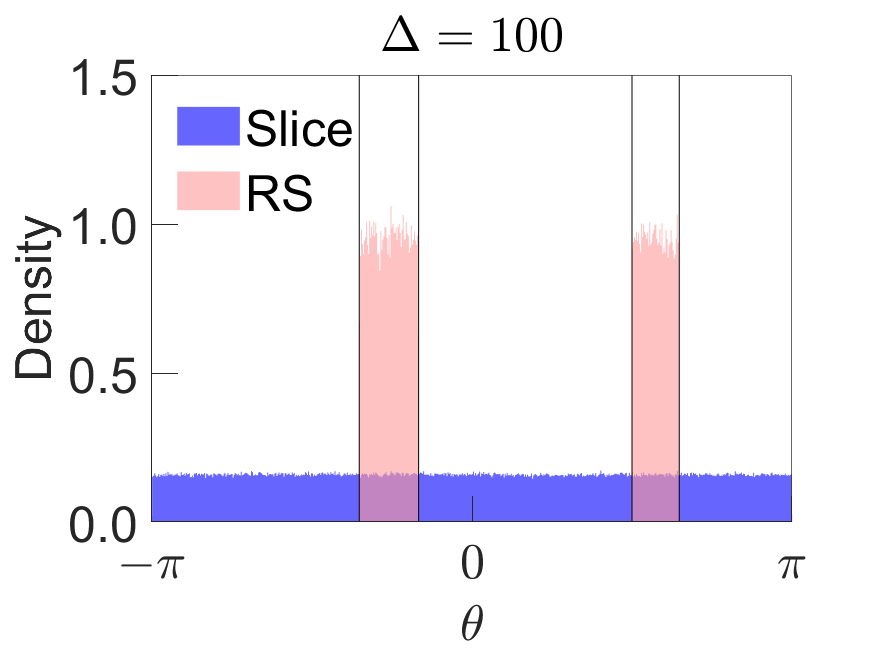}  &  \includegraphics[scale=0.5]{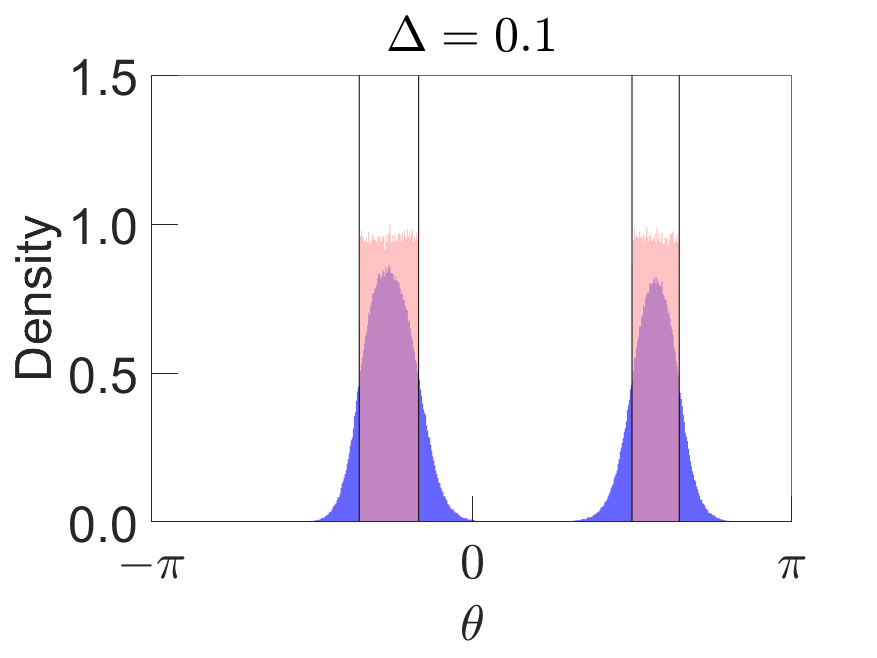} \\
		\includegraphics[scale=0.5]{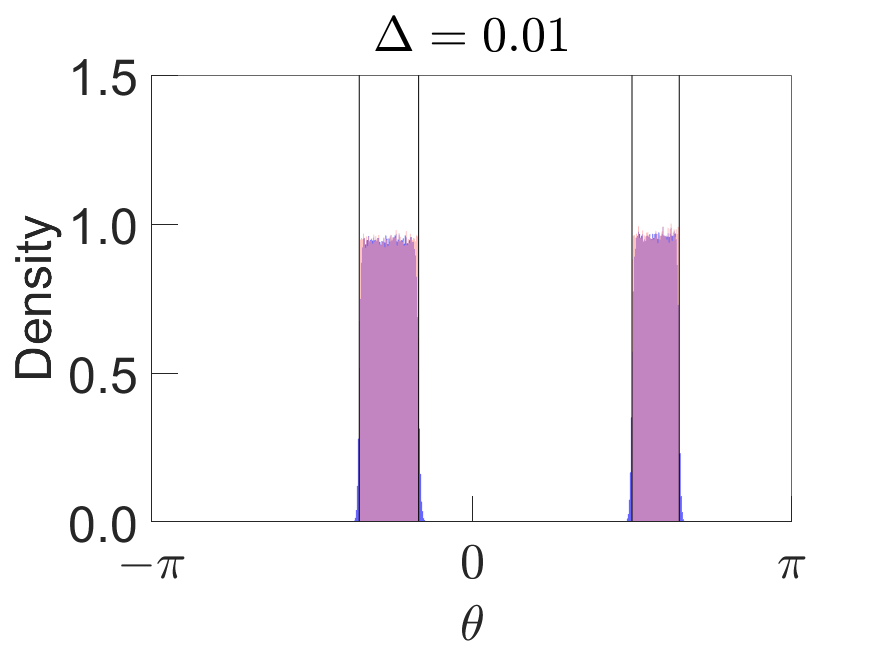}  &  \includegraphics[scale=0.5]{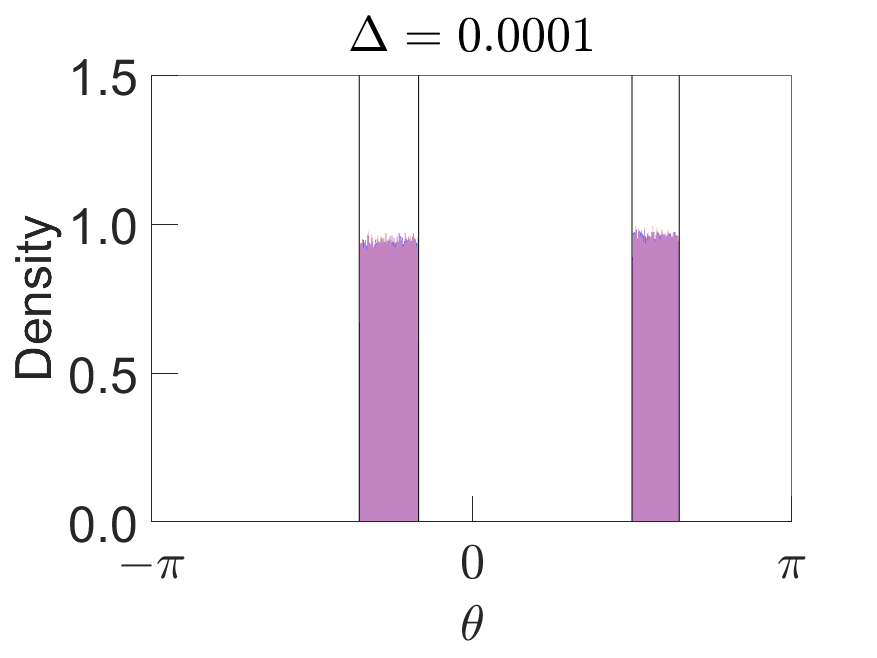} \\      
	\end{tabular}
    \noindent\begin{minipage}{\textwidth}
	\small\textbf{Notes:} `Slice' is distribution of draws obtained using slice sampler before importance sampling; `RS' is distribution after importance sampling. Vertical lines represent bounds of identified set. Parameter values are $\bm{\phi} = (\sigma_{11},\sigma_{21},\sigma_{22})' = (1,-0.5,1)'$ and $\lambda = 0.5$. Based on one~million draws.
    \end{minipage}
\end{figure}

Figure~\ref{fig:illustration_disc} illustrates the sampler under different values of $\Delta$. Even at small values of $\Delta$, the sampler continues to cover the identified set -- and thus generate draws from the target distribution -- despite the identified set being disconnected. In the case where $\Delta = 0.0001$, 55.2~per cent of draws lie within the first interval, which is close to the theoretical probability under the uniform distribution (55.7~per cent). Our sampler therefore appears to adequately mix across the two regions.

The examples here and in the previous section point to the potential for our approach to improve the computational efficiency of posterior sampling under sign restrictions when the restrictions substantially truncate the identified set, including when the identified set is disconnected. To assess whether the approach can deliver on this promise, in the next section we turn to a realistic empirical application.

\section{Empirical Application: Shocks in the Oil Market}
\label{sec:empirical}
To examine the performance of the algorithms empirically, we consider the model of the oil market in \citet{Antolin-Diaz_Rubio-Ramirez_2018} (henceforth, ARR18). This model involves a rich set of identifying restrictions -- sign restrictions, `elasticity' restrictions and narrative restrictions -- that (nonlinearly) constrains all columns of $\mathbf{Q}$. We compare the performance of our sampler against accept-reject sampling when conducting standard Bayesian inference under a uniform prior. We also demonstrate the utility of our sampler when conducting prior-robust Bayesian inference.

\subsection{Model and identifying restrictions}
\label{subsec:model}

The model's endogenous variables are an index of real economic activity ($REA_t$), the growth rate of global oil production ($PROD_t$) and the log of the real price of oil ($RPO_t$). The VAR includes 24~lags and a constant, and is estimated on monthly data from January 1971 to December 2015.\footnote{The data were obtained from the replication files to ARR18.} The reduced-form prior is a diffuse normal-inverse-Wishart distribution, so the posterior is also normal-inverse-Wishart (e.g. \citealt{DelNegro_Schorfheide_2011}).

Let $\mathbf{y}_{t} = (REA_t, PROD_t, RPO_t)'$. The following sign restrictions are imposed on the impact impulse responses:
\begin{equation}\label{eq:signrestrictionsoil}
	\mathbf{A}_{0}^{-1} = 
	\begin{bmatrix}
		+ & - & - \\
		+ & + & - \\
		+ & + & +
	\end{bmatrix}.
\end{equation}
These restrictions imply that the model's three structural shocks can be interpreted as shocks to aggregate demand, oil-specific demand and oil supply, respectively. We also impose the restrictions that $\mathrm{diag}(\mathbf{A}_{0}) = \mathrm{diag}(\mathbf{Q}'\bm{\Sigma}_{tr}^{-1}) \geq \mathbf{0}_{s\times 1}$, which can be viewed as a normalisation on the signs of the structural shocks.

The `elasticity' restrictions restrict the `price elasticity of oil supply' to be less than 0.0258, which \citet{Kilian_Murphy_2012} argue is a credible upper bound based on existing evidence. This `elasticity' is defined as the ratio of the impact response of $PROD_t$ to the impact response of $RPO_t$ following aggregate demand or oil-specific demand shocks, so the restrictions are:\footnote{\cite{Baumeister_Hamilton_2024} demonstrate that ratios of impulse responses cannot be interpreted as structural elasticities; instead, structural elasticities are given by ratios of elements of $\mathbf{A}_{0}$. We impose the same restrictions as in ARR18, who in turn follow \citet{Kilian_Murphy_2012}, to maintain comparability.}
\begin{align}
	\frac{\mathbf{e}_{2,3}'\mathbf{A}_{0}^{-1}\mathbf{e}_{1,3}}{\mathbf{e}_{3,3}'\mathbf{A}_{0}^{-1}\mathbf{e}_{1,3}} &= \frac{\mathbf{e}_{2,3}'\bm{\Sigma}_{tr}\mathbf{q}_{1}}{\mathbf{e}_{3,3}'\bm{\Sigma}_{tr}\mathbf{q}_{1}} \leq 0.0258 \\
	\frac{\mathbf{e}_{2,3}'\mathbf{A}_{0}^{-1}\mathbf{e}_{2,3}}{\mathbf{e}_{3,3}'\mathbf{A}_{0}^{-1}\mathbf{e}_{2,3}} &= \frac{\mathbf{e}_{2,3}'\bm{\Sigma}_{tr}\mathbf{q}_{2}}{\mathbf{e}_{3,3}'\bm{\Sigma}_{tr}\mathbf{q}_{2}} \leq 0.0258.
\end{align}

The narrative restrictions include restrictions on the signs of the structural shocks and their contributions to one-step-ahead forecast errors in selected episodes (i.e.~historical decompositions). The shock-sign restrictions are that the oil supply shock was nonnegative in December 1978, January 1979, September 1980, October 1980, August 1990, December 2002, March 2003 and February 2011, which are months in which narrative accounts suggest that there were unexpected disruptions in oil production.\footnote{ARR18 state that the oil supply shock is negative in these periods, which reflects a convention of referring to supply shocks that lower production as negative. However, given the sign restrictions in Equation~(\ref{eq:signrestrictionsoil}), a \textit{positive} supply shock results in $PROD_t$ decreasing and $RPO_t$ increasing. Hence, although the language that we use to describe the sign of the shock differs, the economic content of the restriction is the same.} These restrictions require that
\begin{equation}
	\varepsilon_{3t} = \mathbf{e}_{3,3}'\mathbf{A}_{0}\mathbf{u}_{t} = (\bm{\Sigma}_{tr}^{-1}\mathbf{u}_{t})\mathbf{q}_{3} \geq 0
\end{equation}
for values of $t$ corresponding to the dates above. The historical-decomposition restrictions include the restriction that the oil supply shock was the `most important contributor' to the observed unexpected movement in $PROD_t$ in these months (i.e. $|H_{2,3,t,t}| \geq \max_{j \neq 3} |H_{2,j,t,t}|$). Finally, in September 1980, October 1980 and August 1990, aggregate demand shocks are required to be the `least important contributor' to the unexpected movement in $RPO_t$ (i.e. $|H_{3,1,t,t}| \leq \min_{j \neq 1} |H_{3,j,t,t}|$).

\subsection{Standard Bayesian inference}
\label{subsec:empiricalstandard}

The goal of our first exercise is to obtain 1,000 draws of $\bm{\phi}$ from its posterior (such that the identified set is nonempty) and $M=$ 1,000 draws of $\mathbf{Q}$ from the uniform distribution over $\mathcal{Q}(\bm{\phi}\mid S)$ at each draw of $\bm{\phi}$, yielding $10^6$ draws of the impulse responses. We do this using accept-reject and our sampler, embedding both samplers in a standard posterior sampler for $\bm{\phi}$.

In implementing the accept-reject sampler in this application, we introduce an additional `sign normalisation' step into Algorithm~\ref{algo:acceptreject} to mechanically impose the restriction  $\mathrm{diag}(\mathbf{A}_{0}) = \mathrm{diag}(\mathbf{Q}'\bm{\Sigma}_{tr}^{-1}) \geq \mathbf{0}_{s\times 1}$.\footnote{Given a draw of $\mathbf{Q}=(\mathbf{q}_1,\ldots,\mathbf{q}_n)$ obtained via Step~1) of Algorithm~\ref{algo:acceptreject}, the sign-normalised draw of $\mathbf{Q}$ has $j$th column $\mathrm{sign}((\bm{\Sigma}_{tr}^{-1}\mathbf{e}_{j,n})'\mathbf{q}_j)\mathbf{q}_j$.} As discussed in \cite{Rubio-Ramirez_Waggoner_Zha_2010}, enforcing these restrictions mechanically (rather than via accept-reject) increases the efficiency of the sampler by a factor of $2^n$. This gives the accept-reject sampler a better chance of outperforming our sampler.

As discussed in Section~\ref{sec:emptyset}, $\mathcal{Q}(\bm{\phi}\mid S)$ may be empty. When using the accept-reject sampler, we make 1,000 unsuccessful attempts to draw $\mathbf{Q}$ before treating $\mathcal{Q}(\bm{\phi}\mid S)$ as empty and redrawing $\bm{\phi}$. When using the slice sampler, if none of the 1,000 draws of $\mathbf{Q}$ satisfy the identifying restrictions, we redraw $\bm{\phi}$. Similar effort is therefore used to determine whether the identified set is nonempty under both approaches. We set $\Delta = 10^{-5}$ when using our sampler and examine alternative choices below. At each draw of $\bm{\phi}$, we initialise the slice sampler by numerically maximising the log target density with $\Delta= 0.1$.

\begin{figure}[h]
	\centering
	\caption{Impulse Responses to Oil Market Shocks -- Standard Bayesian Inference}
	\label{fig:impulseresponses}	
		\includegraphics[scale=1]{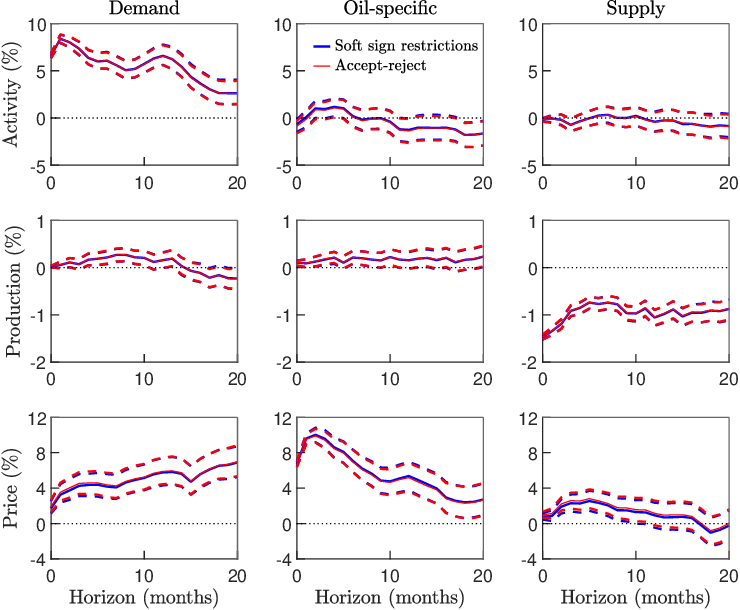}  
    \noindent\begin{minipage}{\textwidth}
	\small\textbf{Notes:} Solid lines are posterior medians and dashed lines are 68~per cent credible intervals.
    \end{minipage}
\end{figure}

Figure~\ref{fig:impulseresponses} summarises the posterior distributions of the impulse responses obtained using the two samplers. The results are very similar.\footnote{The results are also consistent with those in ARR18, despite some differences in the details of the exercise. The results are not directly comparable for two main reasons. First, we use a conditionally uniform prior for $\mathbf{Q}$. Second, following the recommendation in \citet{Giacomini_Kitagawa_Read_2023}, we construct the posterior distribution using the unconditional likelihood rather than the conditional likelihood; this means that the importance-sampling step in ARR18 -- which reweights posterior draws based on the \textit{ex ante} probability that the shocks satisfy the narrative restrictions -- is unnecessary.} The accept-reject sampler takes around 90~hours to generate the desired number of draws from the posterior, whereas our sampler takes only 3~hours. The effective sample size (as a percentage of $M$) from our sampler is around 90~per cent. To adjust for the difference in effective sample size, we compare the number of effective draws per hour. The accept-reject sampler generates approximately 11,000~draws per hour, whereas our sampler generates around 300,000~effective draws per hour. On this basis, our approach is an order of magnitude more efficient than accept-reject sampling, generating roughly 30~times as many effective draws per unit of time.

As discussed in Section~\ref{subsec:softsign}, our use of a self-normalised importance sampler means that there is a finite-sample bias, which vanishes as $M\rightarrow \infty $ (for fixed $\Delta$) or as $\Delta \rightarrow 0$ (for fixed $M$). More generally, the quality of the posterior approximation will also depend on the ability of the slice sampler to explore $f_{\Delta}(\mathbf{Z})$. While Figure~\ref{fig:impulseresponses} suggests that the posterior approximations generated by the samplers are very similar, the computational burden of accept-reject sampling means it will not be practical in general to compare posterior approximations across the two samplers. Moreover, the posterior approximation obtained from the accept-reject sampler may itself be inaccurate due to the necessity of treating the identified set as empty if no draws of $\mathbf{Q}$ satisfy the sign restrictions; intuitively, as demonstrated further below, the accept-reject sampler under-weights values of $\bm{\phi}$ with small $p(\bm{\phi}\mid S)$. 

As a check on the quality of posterior approximations, we therefore recommend comparing the \emph{conditional} posterior distributions of particular parameters of interest across the two samplers. To give an example, Figure~\ref{fig:posteriorcomp} compares the posterior of the impact impulse responses conditional on the posterior mean of $\bm{\phi}$. The two approximations are again very similar.

\begin{figure}
\centering
	\caption{Comparison of Conditional Posterior Approximations Across Samplers}
	\label{fig:posteriorcomp}
	\begin{tabularx}{\textwidth}{X X X}
		\includegraphics[width=\linewidth]{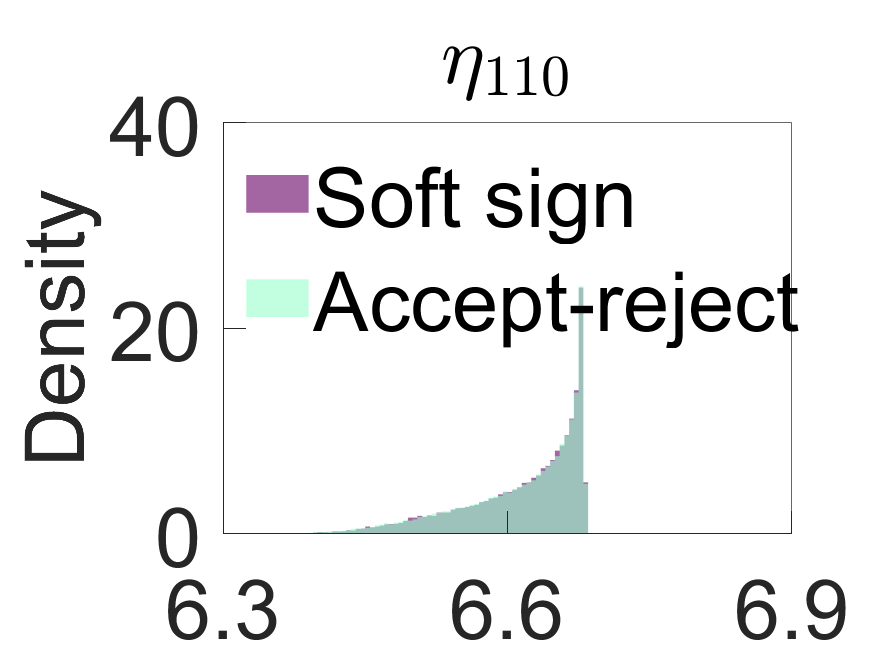}  &   \includegraphics[width=\linewidth]{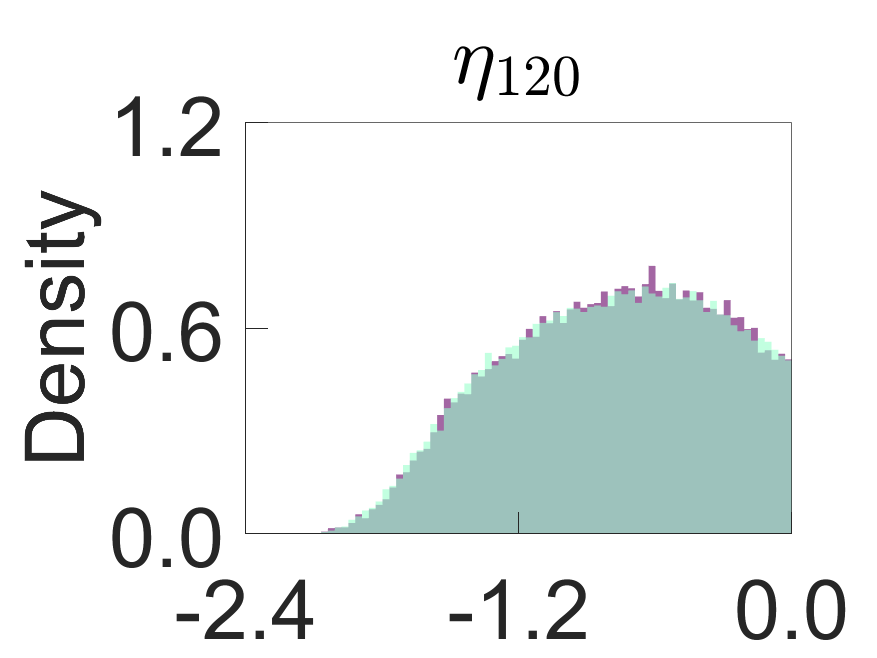}  & \includegraphics[width=\linewidth]{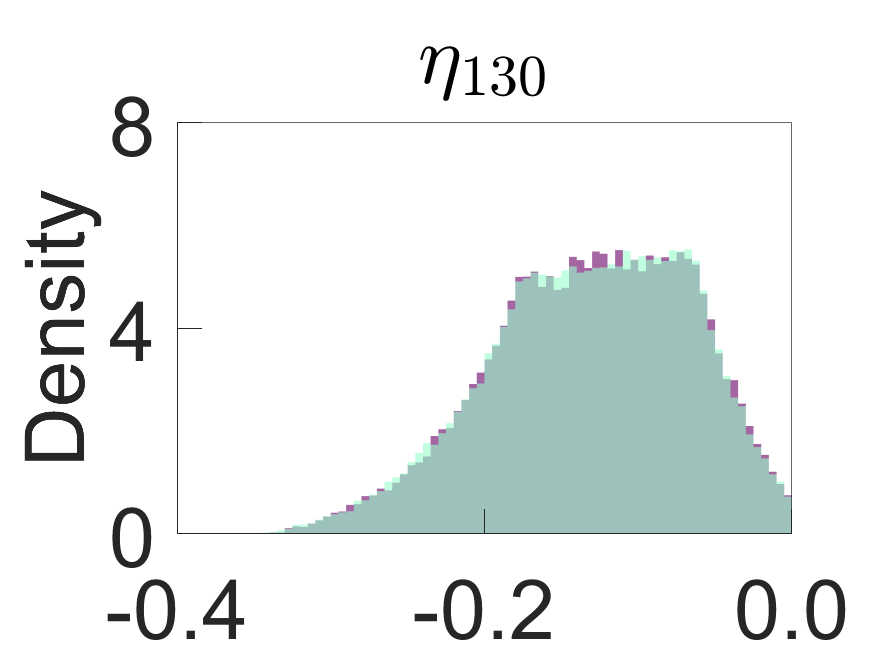}  \\
        \includegraphics[width=\linewidth]{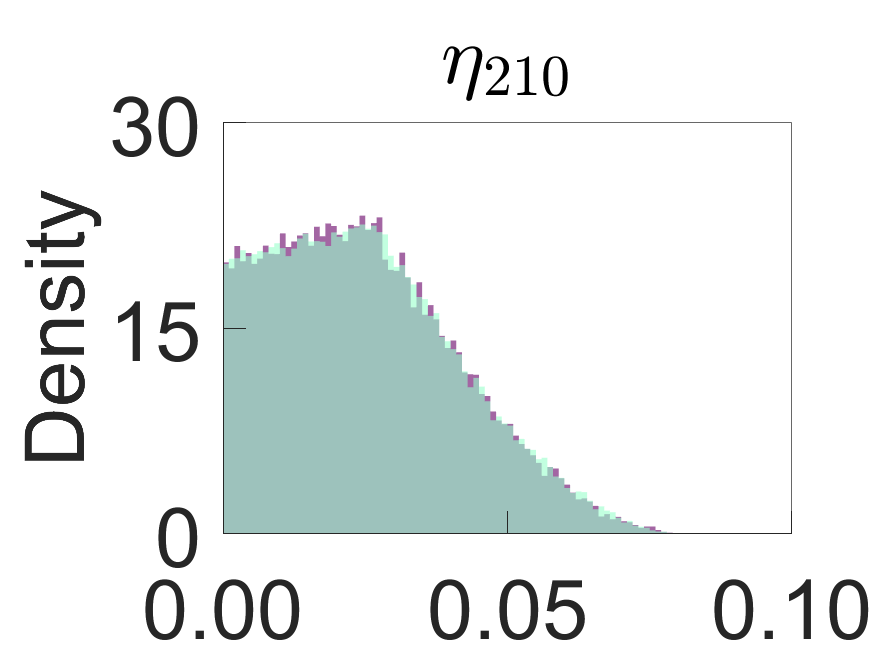}  &   \includegraphics[width=\linewidth]{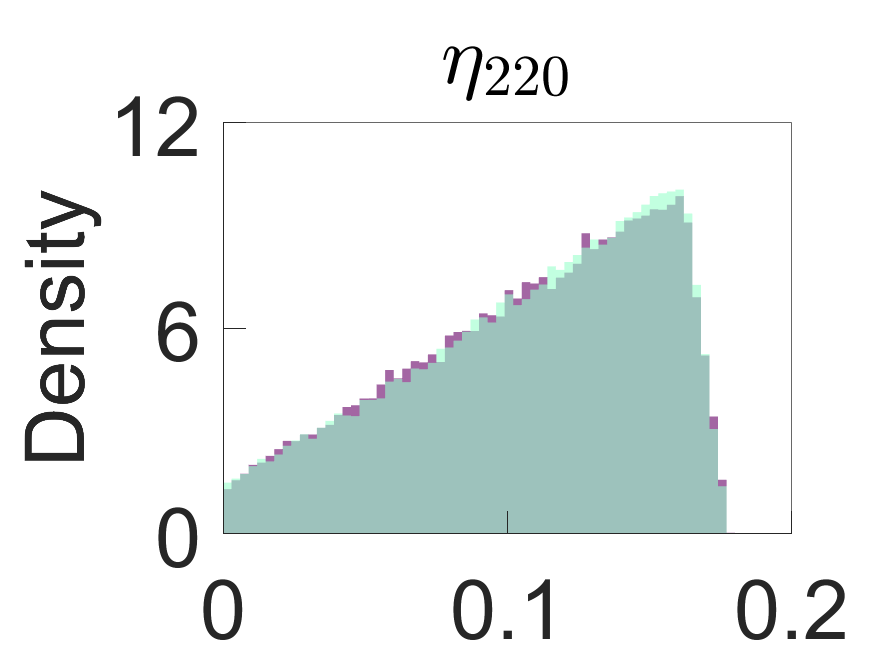}  & \includegraphics[width=\linewidth]{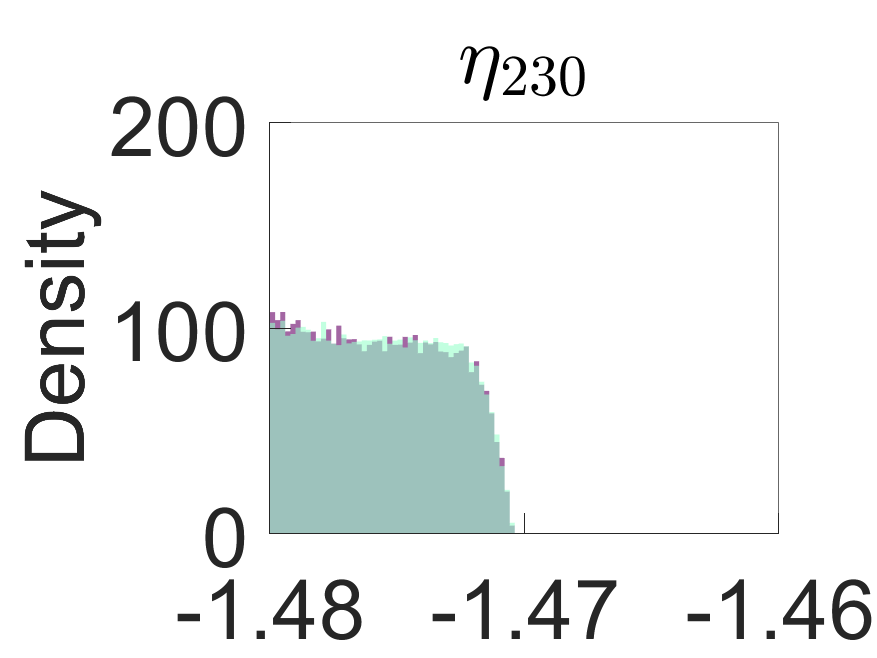}  \\
        \includegraphics[width=\linewidth]{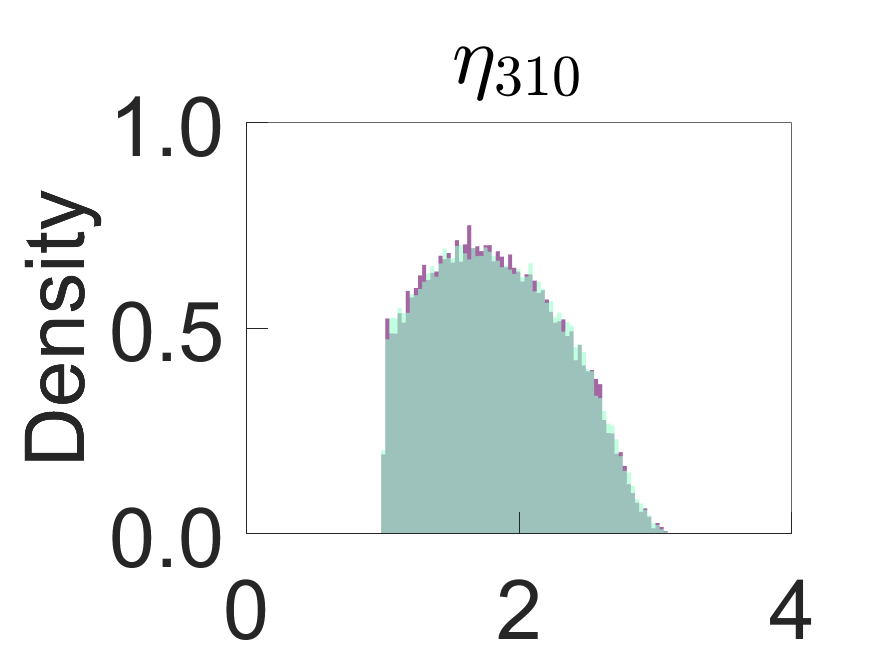}  &   \includegraphics[width=\linewidth]{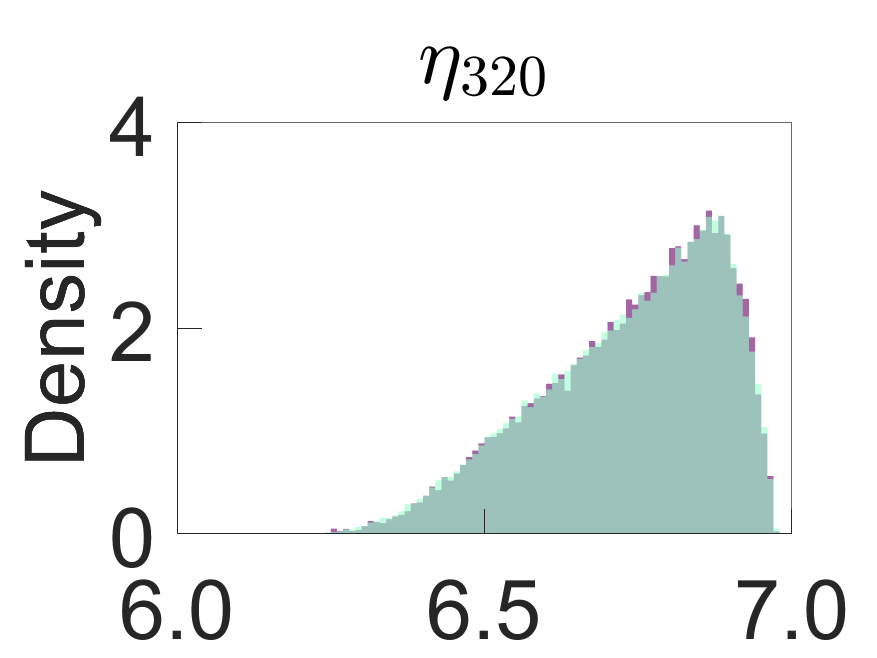}  & \includegraphics[width=\linewidth]{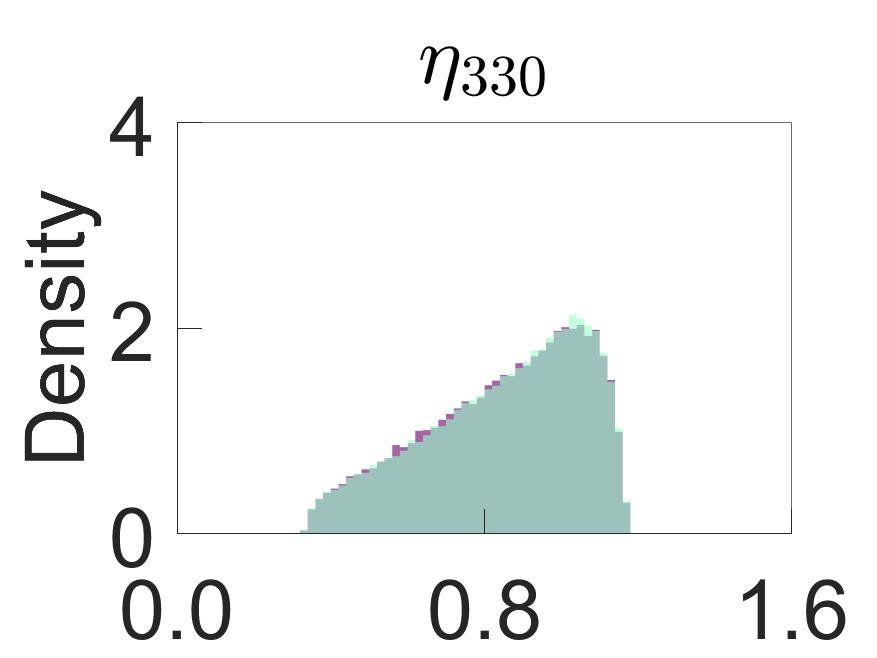}  \\
	\end{tabularx}
    \noindent\begin{minipage}{\textwidth}
	\small\textbf{Notes:} Based on 100,000~draws of $\mathbf{Q}$; conditional on posterior mean of $\bm{\phi}$.
     \end{minipage}  
 \end{figure}

The results in this section have been obtained setting the regularisation parameter $\Delta=10^{-5}$. Table~\ref{tab:performance_oil} compares the computational performance of our sampler at different choices of $\Delta$. For all values of $\Delta$ considered, our sampler generates a larger number of effective draws per hour than the accept-reject sampler. The number of effective draws per hour is broadly similar across a large range of values for $\Delta$ spanning from $\Delta = 10^{-3}$ to $\Delta = 10^{-5}$. Different choices of $\Delta$ also generate quantitatively similar posterior approximations; see the Supplemental Appendix for further details.

\begin{table}
	\caption{Performance of Sampling Algorithms -- Oil Market Model} \label{tab:performance_oil}
	\begin{tabularx}{\textwidth}{X S[table-format=3.1] S[table-format=3.0]}
		\hline
		Algorithm & \multicolumn{1}{c}{Speed (hours)} & \multicolumn{1}{c}{Effective draws per hour ('000)} \\ \hline 
Accept-reject	&	89.9	&	11	\\
Soft sign restrictions:	&		&		\\
$\quad \Delta = 10^{-1}$	&	1.1	&	50	\\
$\quad \Delta = 10^{-2}$	&	1.5	&	270	\\
$\quad \Delta = 10^{-3}$	&	1.9	&	397	\\
$\quad \Delta = 10^{-4}$	&	2.5	&	350	\\
$\quad \Delta = 10^{-5}$	&	2.9	&	310	\\
$\quad \Delta = 10^{-6}$	&	3.2	&	279	\\
     \hline
	\end{tabularx}
     \par\vspace{0.5em}
    \noindent\begin{minipage}{\textwidth}
    \small\textbf{Notes:} $\Delta$ controls penalisation of parameter values that violate sign restrictions. Effective draws per hour rounded to nearest thousand.
    \end{minipage}
\end{table}

\subsection{Robust Bayesian inference}
\label{subsec:empiricalrobust}

In this section we demonstrate the utility of our sampling strategy in implementing the prior-robust Bayesian procedure from \citet{Giacomini_Kitagawa_2021}. This procedure can be used to assess the sensitivity of posterior inferences to the choice of conditional prior for $\mathbf{Q}$ given $\bm{\phi}$, which is not revisable because $\mathbf{Q}$ is set identified.\footnote{In some applications of set-identified SVARs, robust Bayesian methods have revealed that much of the apparent information in the standard Bayesian posterior is contributed by the conditional prior for $\mathbf{Q}$ (e.g. \citealt{Giacomini_Kitagawa_2021,Giacomini_Kitagawa_Read_2022b,Giacomini_Kitagawa_Read_2023,Giacomini_Kitagawa_Read_2026,Read_2024}), though this is not necessarily always the case, particularly when rich sets of identifying restrictions are imposed (e.g. \citealt{Inoue_Kilian_2026}).}

The prior-robust procedure replaces the conditional prior for $\mathbf{Q}$ with the class of all conditional priors that assign probability one to $\mathcal{Q}(\bm{\phi}\mid S)$. Summarising the corresponding class of posteriors requires computing the bounds of the identified set for each parameter of interest at each posterior draw of $\bm{\phi}$. \citet{Giacomini_Kitagawa_2021} suggest approximating the bounds by computing the minimum and maximum of the parameter of interest over many draws of $\mathbf{Q}$ obtained via accept-reject sampling.\footnote{An alternative approach is to obtain the bounds by solving a numerical optimisation problem using, for example, gradient-based methods, but this can be computationally burdensome and convergence to the true bounds is not guaranteed (e.g. \citealt{Amir-Ahmadi_Drautzburg_2021,Giacomini_Kitagawa_2021,Montiel-Olea_Nesbit_2021}).} Given $M$ draws from inside $\mathcal{Q}(\bm{\phi}|S)$, the bounds will be approximated with error that vanishes as $M\rightarrow \infty$. 

\citet{Montiel-Olea_Nesbit_2021} derive theoretical results about the number of draws required to approximate identified sets with a specified degree of accuracy. Ensuring misclassification error less than 5~per cent with probability at least 95~per cent requires over 20,000 draws from $\mathcal{Q}(\bm{\phi}\mid S)$ at each draw of $\bm{\phi}$. Using accept-reject to implement the robust Bayesian procedure with this target level of accuracy would be extremely computationally costly. We therefore turn to our sampler based on soft sign restrictions. See the Supplemental Appendix for details about how we implement our approach in this exercise.

Figure~\ref{fig:priorrobust} summarises the class of posteriors using the `set of posterior medians' and `robust credible interval'. For each impulse response, the set of posterior medians is an interval with lower (upper) bound equal to the posterior median of the lower (upper) bound of the identified set; this interval contains all posterior medians obtainable from the class of posteriors and can be interpreted as an estimator of the identified set.\footnote{If the identified set is nonconvex, the set of posterior medians can be interpreted as an estimator of its convex hull.} The $\alpha$-level robust credible interval is an interval that is assigned at least posterior probability $\alpha$ uniformly under the class of posteriors. 

\begin{figure}[h]
	\centering
	\caption{Impulse Responses to Oil Market Shocks -- Comparison of Standard and Robust Bayesian Inference}
	\label{fig:priorrobust}
		\includegraphics[scale=1]{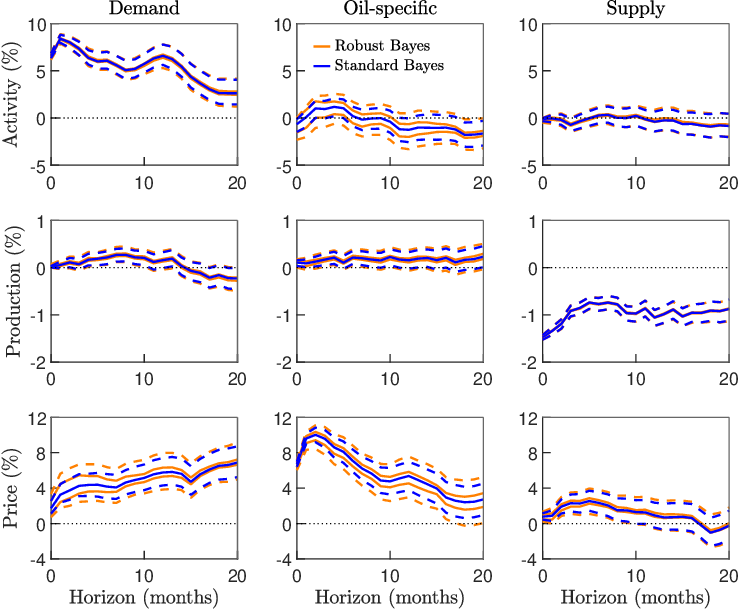}
        \par\vspace{0.5em}
    \noindent\begin{minipage}{\textwidth}
	\small\textbf{Notes:} Solid lines are (sets of) posterior medians and dashed lines are 68~per cent (robust) credible intervals. Results obtained using soft sign restrictions.
    \end{minipage}
\end{figure}

The influence of the conditional prior on the posterior varies across different impulse responses. For the responses to an oil supply shock, the set of posterior medians tends to be narrow, and the robust credible intervals are similar to the standard credible intervals; the responses to an oil supply shock are tightly identified and the conditional prior has little influence on posterior inferences about these responses. The responses of oil prices to demand-side shocks and the response of activity to an oil-specific demand shock are less-tightly identified and the conditional prior contributes more of the apparent information in the posterior.\footnote{To quantify the influence of the conditional prior, the Supplemental Appendix reports the `prior informativeness' statistic from \citet{Giacomini_Kitagawa_2021}. The prior informativeness statistic is smaller for the responses to oil supply shocks than for other shocks.} Nevertheless, even in these cases, the robust credible intervals tend to exclude zero in the same cases where the standard credible intervals exclude zero. 

Overall, this exercise suggests that inferences about the effects of shocks in the oil market obtained under this rich set of identifying restrictions are insensitive to the choice of conditional prior. Importantly, the large number of draws of $\mathbf{Q}$ used to approximate identified sets and the guarantee on approximation accuracy from \citet{Montiel-Olea_Nesbit_2021} mean that this apparent robustness is unlikely to be an artifact of approximation error. These results complement exercises in \citet{Inoue_Kilian_2026}, who highlight applications where posterior inferences do not appear to be driven by the uniform prior for $\mathbf{Q}$.

\subsection{Empty identified sets}
\label{subsec:empiricalempty}

As discussed in Section~\ref{sec:emptyset}, $\mathcal{Q}(\bm{\phi}\mid S)$ may be empty. This section examines the ability of the two samplers to determine whether $\mathcal{Q}(\bm{\phi}\mid S)$ is nonempty. To do this, we compute the `posterior plausibility' of the identifying restrictions, which is the posterior probability (under the posterior for $\bm{\phi}$) that $\mathcal{Q}(\bm{\phi}\mid S)$ is nonempty; the posterior plausibility can be used to quantify the compatibility of the identifying restrictions with the joint distribution of the data, which is indexed by $\bm{\phi}$ (\citealt{Giacomini_Kitagawa_2021,Giacomini_Kitagawa_Read_2022a}).\footnote{The prior and posterior plausibility are also inputs into Bayes factors for comparing models with different set-identifying restrictions (\citealt{Amir-Ahmadi_Drautzburg_2021,Giacomini_Kitagawa_Volpicella_2022}).}

In Section~\ref{subsec:empiricalstandard}, we treat $\mathcal{Q}(\bm{\phi}\mid S)$ as empty at a draw of $\bm{\phi}$ when none of the $M=1,000$ candidate draws of $\mathbf{Q}$ satisfy the identifying restrictions. Based on the accept-reject sampler, the posterior plausibility is less than 1~per cent, suggesting that the restrictions are incompatible with the data. In contrast, the posterior plausibility is around 22~per cent when using our sampler. Our sampler is better able to determine when $\mathcal{Q}(\bm{\phi}\mid S)$ is nonempty, because it tends to drift towards the identified set. An implication of this difference is that the posterior approximation obtained using the accept-reject sampler tends to underweight values of $\bm{\phi}$ with small $p(\bm{\phi}\mid S)$.\footnote{In the robust Bayesian exercise of Section~\ref{subsec:empiricalrobust}, where we use around 23,000 draws of $\mathbf{Q}$, the posterior plausibility increases to around 26~per cent.} 

\subsection{Additional application: US monetary policy shocks}
\label{subsec:empiricalmonetarypolicy}

Supplemental Appendix~C documents the performance of our sampler using the model of US monetary policy from ARR18. The model contains six variables and identifies a monetary policy shock using an extensive set of sign restrictions on impulse responses and narrative restrictions related to eight historical episodes. Obtaining 1,000 draws of $\mathbf{Q}$ at 1,000~draws of $\bm{\phi}$ takes 21.8~hours using the accept-reject sampler, but only 1.1~hours using our approach. The average effective sample size from our sampler is about 83~per cent, and our approach generates around 780,000~effective draws per hour, compared with 45,000~effective draws per hour for the accept-reject sampler. The two samplers also generate similar posterior approximations. We conclude that our approach continues to perform favourably in a larger model.

\section{Conclusion}
\label{sec:conclusion}

We develop an approach to posterior sampling in sign-restricted SVARs under a uniform prior for the orthonormal matrix. This approach can also be used when conducting prior-robust Bayesian inference. The approach samples from a target density that smoothly penalises parameter values violating the identifying restrictions, which allows us to apply MCMC methods, before applying an importance sampling step to obtain draws from the hard-truncated target density. Our approach is broadly applicable under a wide range of identifying restrictions, including elasticity and narrative restrictions. We provide evidence that our approach is more computationally efficient than brute-force accept-reject sampling when the identified set for the orthonormal matrix is assigned small measure under the uniform prior. It is therefore particularly useful when rich sets of identifying restrictions are imposed.

A number of potentially fruitful avenues for further work could build on our approach. It may be possible to improve the efficiency of our approach by using alternative MCMC samplers that exploit local information about the shape of the smoothed target density (e.g.~gradients); the regularised constraints that we use allow us in principle to construct such approximations.\footnote{Extensions of the slice sampler can use local information about the shape of the target density, such as by using local quadratic approximations (\citealt{Neal_2003}). } It is likely that the use of alternative MCMC samplers would be necessary in large models, since the slice sampler can become inefficient in high dimensions. A promising candidate is the elliptical slice sampler, proposed in \cite{Murray_Adams_MacKay_2010}; \cite{Arias_Rubio-Ramirez_Shin_2026} use an elliptical slice sampler to implement their Gibbs sampler in relatively large SVARs. It would be valuable to extend our approach to allow for zero restrictions (e.g. \nocite{Arias_Rubio-Ramirez_Waggoner_2018}Arias \emph{et al} 2018; \nocite{Giacomini_Kitagawa_2021}Giacomini and Kitagawa 2021a; \nocite{Read_2022}Read 2022), which would further broaden its applicability. Soft sign restrictions could also be used within Sequential Monte Carlo (e.g. \citealt{Herbst_Schorfheide_2015}) by sampling from a sequence of target densities with decreasing $\Delta$. It may also be useful to extend our sampling strategy to sign-restricted SVARs where a prior is specified directly over the structural parameters (e.g. \citealt{Baumeister_Hamilton_2015}). The idea of smoothing hard constraints on the parameter space could also be beneficial when conducting inference in other set-identified models where parameter regions are truncated by inequality restrictions.

\appendix
\renewcommand{\theequation}{A\arabic{equation}}
\renewcommand{\thesection}{A}
\renewcommand{\thefigure}{A\arabic{figure}}
\setcounter{equation}{0}
\setcounter{figure}{0}
\section{Appendix}
\label{app:proofs}
\subsection{Proof of Proposition~\ref{prop:conv}}

\begin{proof}
	The assumption on $T$ ensures that $\mathbb{E}_f[T(\mathbf{Z})]$ and $\mathbb{E}_\Delta[T(\mathbf{Z})]$ exist. We can write
	\begin{equation}\label{eq:expectationinequality}
		\left|\mathbb{E}_f[T(\mathbf{Z})]-\mathbb{E}_\Delta[T(\mathbf{Z})]\right|\leq \int |T(\mathbf{Z})|\left\vert f(\mathbf{Z} | Q(\mathbf{Z}) \in \mathcal{Q}(\bm{\phi}\mid S))- f_{\Delta}(\mathbf{Z})\right\vert d\mathbf{Z}.
	\end{equation}
	Denote the normalising constants for $f(\mathbf{Z} \mid  Q(\mathbf{Z}) \in \mathcal{Q}(\bm{\phi}\mid S))$ and $f_{\Delta}(\mathbf{Z})$ by $C_f$ and $C_\Delta$, respectively. Without loss of generality, assume $s=1$ so there is a single sign restriction $S(\bm{\phi},\mathbf{Q}) \geq 0$. The right-hand side of Equation~(\ref{eq:expectationinequality}) can be written as
	\begin{equation}\label{eq:expectationinequality2}
		\begin{split}
			\begin{array} {rcl}
		C_f^{-1}\int |T(\mathbf{Z})|f_Z(\mathbf{Z})\left\vert \mathbbm{1}(S(\bm{\phi},Q(\mathbf{Z}))\geq0)- \frac{1}{1+\exp(-S(\bm{\phi},Q(\mathbf{Z}))/\Delta)}\right\vert d\mathbf{Z} \\
		+\left|\frac{C_f-C_\Delta}{C_f}\right|\int |T(\mathbf{Z})| f_{\Delta}(\mathbf{Z})d\mathbf{Z}.
			\end{array}
		\end{split}
	\end{equation}
	Under Assumption \ref{ass1},
	\begin{equation}
		\left\vert \mathbbm{1}(S(\bm{\phi},Q(\mathbf{Z}))\geq0)- \frac{1}{1+\exp(-S(\bm{\phi},Q(\mathbf{\mathbf{Z}}))/\Delta)}\right\vert\leq K.
	\end{equation}
	The first term in Equation~(\ref{eq:expectationinequality2}) therefore goes to zero as $\Delta \rightarrow 0$ by the monotone convergence theorem. In the second term of Equation~(\ref{eq:expectationinequality2}),
	\begin{equation}
		|C_f-C_\Delta|\leq \int f_Z(\mathbf{Z})\left\vert \mathbbm{1}(S(\bm{\phi},Q(\mathbf{Z}))\geq0)- \frac{1}{1+\exp(-S(\bm{\phi},Q(\mathbf{Z}))/\Delta)}\right\vert d\mathbf{Z},
	\end{equation}
	which similarly goes to zero. 
\end{proof}

\subsection{Proof of Proposition~\ref{prop:bounded_weights}}

\begin{proof}
If $Q(\mathbf{Z})\notin \mathcal Q(\bm{\phi}\mid S)$, then $\mathbbm{1}\left(Q(\mathbf{Z})\in \mathcal Q(\bm{\phi}\mid S)\right)=0$ and $\tilde w_{\Delta}(\mathbf{Z})=0$, so the lower bound on $\tilde w_{\Delta}(\mathbf{Z})$ holds trivially.

Now suppose $Q(\mathbf{Z})\in \mathcal Q(\bm{\phi}\mid S)$. By the definition of the identified set $\mathcal Q(\bm{\phi}\mid S)$, each inequality restriction is satisfied, i.e.
$S^{(\ell)}(\bm{\phi},Q(\mathbf{Z}))\ge 0$ for all $\ell=1,\ldots,s$. We proceed under the maintained assumption that $\Lambda(x,\Delta)$ is the logistic function (\ref{eq:logistic}). When $x\geq 0$, we have $\exp(-x/\Delta)\in(0,1]$, which implies
\[
\Lambda(x,\Delta)\in\left[\frac{1}{2},\,1\right).
\]
Therefore, for each $\ell$,
\[
\Lambda\!\left(S^{(\ell)}(\bm{\phi},Q(\mathbf{Z})),\Delta\right)\ \ge\ \frac{1}{2},
\]
and hence
\[
\prod_{\ell=1}^s \Lambda\!\left(S^{(\ell)}(\bm{\phi},Q(\mathbf{Z})),\Delta\right)
\ \ge\ \left(\frac{1}{2}\right)^s.
\]
Using that $\mathbbm{1}\left(Q(\mathbf{Z})\in \mathcal Q(\bm{\phi}\mid S)\right)=1$, we obtain 
\[
\tilde w_{\Delta}(\mathbf{Z})
=
\frac{1}{\prod_{\ell=1}^s \Lambda\!\left(S^{(\ell)}(\bm{\phi},Q(\mathbf{Z})),\Delta\right)}
\ \le\
\frac{1}{(1/2)^s}
=
2^s.
\]
Combining the two cases yields $0\leq \tilde w_{\Delta}(\mathbf{Z})\le 2^s$ for all $\mathbf{Z}$. 
\end{proof}

\newpage
\putbib
\end{bibunit}

\clearpage

\makeatletter
\gdef\@thanks{}
\renewcommand{\thefootnote}{\fnsymbol{footnote}} 
\setcounter{footnote}{1}                          
\makeatother

  \title{Supplemental Appendix for `Fast Posterior Sampling in Tightly Identified SVARs Using `Soft' Sign Restrictions'}
  \author{Matthew Read\thanks{Economic Research Department, Reserve Bank of Australia. email: readm@rba.gov.au}\\
Dan Zhu\thanks{Department of Econometrics and Business Statistics, Monash University. email: dan.zhu@monash.edu.au}} 
  \date{\today}               
  \maketitle

\renewcommand{\thefootnote}{\arabic{footnote}}
\setcounter{footnote}{1}

\begin{bibunit}

This Supplemental Appendix contains additional material related to the bivariate model and empirical applications. Appendix~\ref{app:bivariate} derives identified sets in the bivariate model of Section~4. Appendix~\ref{app:empirical} contains information related to the oil market model of Section~5, including additional results and details about how we use our sampler to implement the prior-robust Bayesian approach to inference. Appendix~\ref{app:monetarypolicy} contains material related to the model of US monetary policy, which is briefly discussed in Section~5.5 of the main text.

\renewcommand{\theequation}{S1-\arabic{equation}}
\renewcommand{\thesection}{S1}
\renewcommand{\thefigure}{S1-\arabic{figure}}
\setcounter{equation}{0}
\setcounter{figure}{0}
\section{Bivariate Model -- Additional Details}
\label{app:bivariate}

Let $\mathbf{y}_{t} = (p_{t},q_{t})'$ contain log price and quantity. Assume $\mathbf{y}_t$ follows the SVAR($0$) $\mathbf{y}_t = \mathbf{A}_0^{-1}\bm{\varepsilon}_t = \bm{\Sigma}_{tr}\mathbf{Q}\mathbf{u}_t$ and let $\bm{\phi}=\mathrm{vech}(\bm{\Sigma}_{tr}) = (\sigma_{11},\sigma_{21},\sigma_{22})'$. The space of $2\times 2$ orthonormal matrices $\mathcal{O}(2)$ can be represented as
\begin{equation}
	\mathcal{O}(2) = \left\{
	\begin{bmatrix}
		\cos\theta & -\sin\theta \\
		\sin\theta & \cos\theta 
	\end{bmatrix}
	\right\} \cup
	\left\{
	\begin{bmatrix}
		\cos\theta & \sin\theta \\
		\sin\theta & -\cos\theta 
	\end{bmatrix}
	\right\},
\end{equation}
where $\theta \in [-\pi,\pi]$. The first set is the set of `rotation' matrices and the second is the set of `reflection' matrices.

\subsection{Connected identified set}

Consider imposing the following sign restrictions: 
\begin{equation}
	\mathbf{A}_{0}^{-1} = 
	\begin{bmatrix}
		+ & + \\
		- & +
	\end{bmatrix}.
\end{equation}
These restrictions require $\mathbf{e}_{1,2}'\mathbf{A}_{0}^{-1}\mathbf{e}_{1,2} = \mathbf{e}_{1,2}'\bm{\Sigma}_{tr}\mathbf{q}_{1} \geq 0$, $\mathbf{e}_{2,2}'\mathbf{A}_{0}^{-1}\mathbf{e}_{1,2} = \mathbf{e}_{2,2}'\bm{\Sigma}_{tr}\mathbf{q}_{1} \leq 0$, $\mathbf{e}_{1,2}'\mathbf{A}_{0}^{-1}\mathbf{e}_{2,2} = \mathbf{e}_{1,2}'\bm{\Sigma}_{tr}\mathbf{q}_{2} \geq 0$ and $\mathbf{e}_{2,2}'\mathbf{A}_{0}^{-1}\mathbf{e}_{2,2} = \mathbf{e}_{2,2}'\bm{\Sigma}_{tr}\mathbf{q}_{2} \geq 0$.\footnote{Sign normalisation restrictions on $\mathrm{diag}(\mathbf{A}_0)$ are redundant given these sign restrictions.} Assuming that $\sigma_{21} < 0$, it can be shown that the identified set for $\theta$ under these sign restrictions is:\footnote{For derivations, see \cite{Baumeister_Hamilton_2015} or \cite{Read_2022_rdp}.}
\begin{equation}\label{eq:is_sign}
	IS_{\theta}(\bm{\phi}\mid S) = \left[\arctan\left(\frac{\sigma_{22}}{\sigma_{21}}\right),0\right].
\end{equation}
Consider augmenting the sign restrictions with an upper bound on the price elasticity of supply $ \omega(\bm{\phi},\mathbf{Q}) \leq \bar{\omega}$ with $\bar{\omega} \geq 0$. The price elasticity of supply in this model is equivalent to the impulse response of $q_t$ to a demand shock that raises $p_t$ by one unit:\footnote{The price elasticity of supply can equivalently be obtained as the coefficient on $p_t$ in the first structural equation after normalising the coefficient on $q_t$ to unity:
\begin{equation*}
    \omega(\bm{\phi},\mathbf{Q}) = -\mathbf{e}_{1,2}'\mathbf{A}_{0}\mathbf{e}_{1,2}/\mathbf{e}_{1,2}'\mathbf{A}_{0}\mathbf{e}_{2,2} = -(\bm{\Sigma}_{tr}^{-1}\mathbf{e}_{1,2})'\mathbf{q}_1/(\bm{\Sigma}_{tr}^{-1}\mathbf{e}_{2,2})'\mathbf{q}_1.
\end{equation*}}
\begin{equation}\label{eq:elasticityrestriction}
    \omega(\bm{\phi},\mathbf{Q}) = \mathbf{e}_{2,2}'\mathbf{A}_{0}^{-1}\mathbf{e}_{2,2}/\mathbf{e}_{1,2}\mathbf{A}_{0}^{-1}\mathbf{e}_{2,2} = \mathbf{e}_{2,2}'\bm{\Sigma}_{tr}\mathbf{q}_{2}/\mathbf{e}_{1,2}'\bm{\Sigma}_{tr}\mathbf{q}_{2}.
\end{equation}
Since $\mathbf{e}_{1,2}\mathbf{A}_{0}^{-1}\mathbf{e}_{2,2} \geq 0$, the elasticity restriction is equivalent to
\begin{equation}
    \mathbf{e}_{2,2}'\bm{\Sigma}_{tr}\mathbf{q}_{2} \leq \bar{\omega}\mathbf{e}_{1,2}'\bm{\Sigma}_{tr}\mathbf{q}_{2}.
\end{equation}
It can be shown directly that this restriction implies that
\begin{equation}
    \theta \leq \mathrm{arccot}\left(\frac{\sigma_{21}}{\sigma_{22}}-\frac{\sigma_{11}}{\sigma_{22}}\bar{\omega}\right).
\end{equation}
Intersecting this inequality with the identified set in (\ref{eq:is_sign}) yields:
\begin{equation}
	IS_{\theta}(\bm{\phi}\mid S) = \left[\arctan\left(\frac{\sigma_{22}}{\sigma_{21}}\right),\mathrm{arccot}\left(\frac{\sigma_{21}}{\sigma_{22}}-\frac{\sigma_{11}}{\sigma_{22}}\bar{\omega}\right)\right].
\end{equation}

\subsection{Disconnected identified set}

Consider imposing the restriction that $\eta_{120}(\bm{\phi},\mathbf{Q}) =  \mathbf{e}_{1,2}'\bm{\Sigma}_{tr}\mathbf{q}_{2} \geq \lambda$ for $0 \leq \lambda \leq \sigma_{11}$ (when $\lambda > \sigma_{11}$, $IS_{\theta}(\bm{\phi}\mid S) = \{\emptyset\}$ at any value of $\bm{\phi}$). All other impulse responses are unrestricted and we impose the sign normalisation $\mathrm{diag}(\mathbf{A}_0)=\mathrm{diag}(\mathbf{Q}'\bm{\Sigma}_{tr}^{-1}) \geq \mathbf{0}_{n\times 1}$.\footnote{When $\lambda = 0$, this example coincides with Example~B.5 in \citet{Giacomini_Kitagawa_2021sup}.} 

Under these restrictions, the identified set for $\theta$ is characterised by
\begin{multline}
    \left\{\theta: \sigma_{22}\cos\theta \geq \sigma_{21}\sin\theta, \sigma_{11}\cos\theta \geq 0,-\sigma_{11}\sin\theta \geq \lambda \right\} \\
    \cup \left\{\theta : \sigma_{22}\cos\theta \geq \sigma_{21}\sin\theta, \sigma_{11}\cos\theta \leq 0, \sigma_{11}\sin\theta \geq \lambda\right\}.
\end{multline}
Assuming $\sigma_{21} < 0$, this set is equivalent to
\begin{multline}
    \left\{\theta: \tan\theta \geq \sigma_{22}/\sigma_{21}, \cos\theta > 0, \sin\theta \leq -\lambda/\sigma_{11}\right\} \\
    \cup \left\{\theta : \tan\theta \leq \sigma_{22}/\sigma_{21}, \cos\theta < 0, \sin\theta \geq \lambda/\sigma_{11}\right\} \cup \{\pi/2\}.
\end{multline}
The first set implicitly restricts $\theta$ to lie in $(-\pi/2,0]$, where $\sin(\theta)$ is increasing. The first set is therefore equivalent to $[\arctan(\sigma_{22}/\sigma_{21}),\arcsin(-\lambda/\sigma_{11})]$. The second set implicitly restricts $\theta$ to lie in $[\pi/2,\pi]$, where $\sin(\theta)$ is decreasing. For $\theta \in [\pi/2,\pi]$, $\tan\theta \leq \sigma_{22}/\sigma_{21}$ if $\theta \leq \pi + \arctan(\sigma_{22}/\sigma_{21})$ and $\sin\theta \geq \lambda/\sigma_{11}$ if $\theta \leq \pi - \arcsin(\lambda/\sigma_{11})$. The second set (in union with $\{\pi/2\}$) is therefore equivalent to $[\pi/2,\min\{\pi-\arcsin(\lambda/\sigma_{11}),\pi+\arctan(\sigma_{22}/\sigma_{21})\}$. Putting this together,
\begin{equation}
	\begin{split}
		\begin{array} {rcl}
	IS_{\theta}(\bm{\phi}\mid S) = \left[\arctan\left(\frac{\sigma_{22}}{\sigma_{21}}\right),\arcsin\left(-\frac{\lambda}{\sigma_{11}}\right)\right] \cup \\
	\left[\frac{\pi}{2}, \min\left\{\pi - \arcsin\left(\frac{\lambda}{\sigma_{11}}\right),\pi + \arctan\left(\frac{\sigma_{22}}{\sigma_{21}}\right)\right\} \right],
		\end{array}
	\end{split}
\end{equation}
which is the union of two disconnected intervals.\footnote{This expression assumes $\lambda/\sigma_{11} \leq \sigma_{22}/\sqrt{\sigma_{22}^2 + \sigma_{21}^2}$, otherwise the first interval is empty.}

\renewcommand{\theequation}{S2-\arabic{equation}}
\renewcommand{\thesection}{S2}
\renewcommand{\thefigure}{S2-\arabic{figure}}
\renewcommand{\thetable}{S2-\arabic{figure}}
\setcounter{equation}{0}
\setcounter{figure}{0}
\section{Oil Market Model -- Additional Details}
\label{app:empirical}

\subsection{Comparisons of Posterior Approximations}
\label{app:posteriorapproximations}

Section~5.2 of the main text compares the results obtained using the accept-reject sampler against those from our sampler with $\Delta = 10^{-5}$. Figure~\ref{fig:impulseresponses_Deltas} shows that the approximation of the posterior, as summarised by posterior medians and 68~per cent credible intervals, is substantively unchanged under alternative choices of $\Delta \in \{10^{-4},10^{-6}\}$.

\begin{figure}[h]
	\centering
	\caption{Impulse Responses to Oil Market Shocks -- Comparison of Results with Different $\Delta$}
	\label{fig:impulseresponses_Deltas}	
		\includegraphics[scale=1]{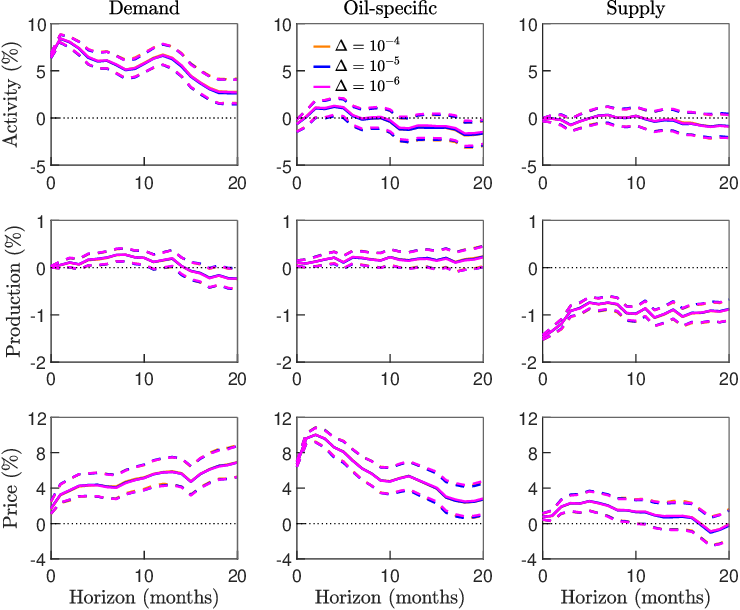}  
    \noindent\begin{minipage}{\textwidth}
	\small\textbf{Notes:} Solid lines are posterior medians and dashed lines are 68~per cent credible intervals.
    \end{minipage}
\end{figure}

\subsection{Approximating identified sets}
\label{app:identifiedsets}

This section provides additional details about how we approximate identified sets when implementing the prior-robust Bayesian procedure from \cite{Giacomini_Kitagawa_2021} (Section~5.3 in the main text).

According to the upper bound in Theorem~3 of \citet{Montiel-Olea_Nesbit_2021}, the number of draws $M$ required from inside the identified set to guarantee a misclassification error less than $\epsilon$ occurs with probability at least $1-\delta$ is 
\begin{equation}
\min \left\{2 d \ln (2d/\delta), \exp(1)(2d + \ln(1/\delta))\right\}/\epsilon,
\end{equation}
where $d$ is the dimension of the parameter region being approximated (i.e. the number of impulse responses). Setting $d = 3\times 3\times 21 = 189$ and $\epsilon = \delta = 0.05$ yields $M=20,713$. This many draws is also consistent with other combinations of $\delta$ and $\epsilon$. Following the recommendation in \citet{Montiel-Olea_Nesbit_2021}, Figure~\ref{fig:isodraw} plots the `iso-draw curve', which traces out combinations of $(\epsilon,\delta)$ that support the target value of $M$.

\begin{figure}[h]
	\centering
	\caption{Iso-draw Curve}
	\label{fig:isodraw}
		\includegraphics[scale=1]{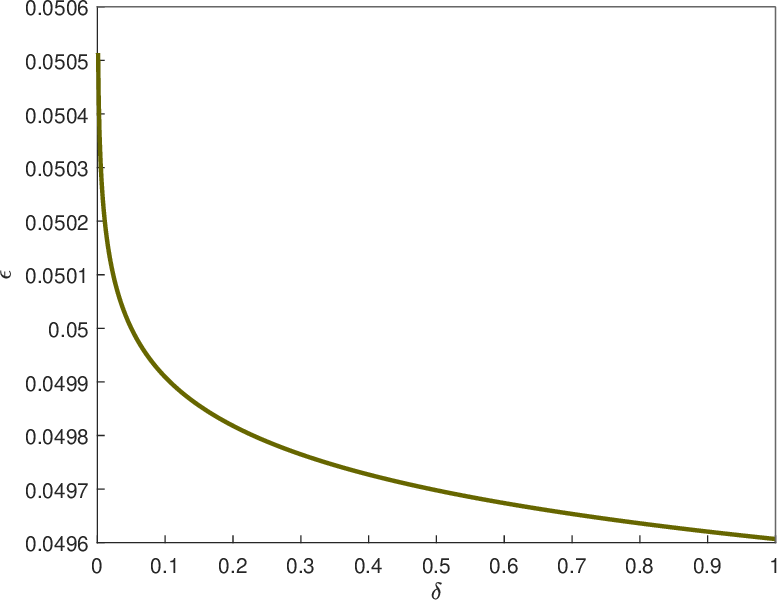}  
        \par\vspace{0.5em}
    \noindent\begin{minipage}{\textwidth}
	\small\textbf{Notes:} Iso-draw curve plots pairs $(\delta,\epsilon)$ such that $M=20,713$ draws from inside the identified set guarantees misclassification error less than $\epsilon$ with probability at least $1-\delta$.
    \end{minipage}
\end{figure}

To obtain $M$ effective draws from inside the identified set, we need to obtain more than $M$ draws from the smoothed target density because some draws will violate the identifying restrictions. We therefore gross up $M$ using the average effective sample size from the exercise in Section~5.2. In that exercise, the average effective sample size was 89~per cent, so obtaining approximately $M$ effective draws on average across the draws of $\bm{\phi}$ requires approximately $M/0.9 \approx 23,000$ draws of $\mathbf{Q}$ from the smoothed target density.

\subsection{Prior informativeness}
\label{app:priorinf}

Section 5.3.1 qualitatively discusses the role of the conditional prior for $\mathbf{Q}$ in driving posterior inferences about the effects of oil market shocks by comparing the results obtained under the standard and robust Bayesian approaches to inference. \citet{Giacomini_Kitagawa_2021} suggest quantifying the influence of the conditional prior in driving posterior inferences using the `prior informativeness statistic'. This is the amount by which the selection of a single conditional prior narrows the credible intervals relative to the robust credible intervals:

\begin{equation}
	\text{Prior informativeness} = 1 - \frac{\text{width of credible interval}}{\text{width of robust credible interval}}.
\end{equation}

When the parameter of interest is point identified, the credible and robust credible intervals coincide, the conditional prior has no influence on posterior inference, and the prior informativeness statistic is zero. Larger values indicate that the conditional prior contributes an increasing proportion of the information contained in the posterior.

\begin{table}[h]
	\caption{Prior Informativeness Statistic -- Impulse Responses to Oil Market Shocks} \label{tab:priorinf}
\begin{tabularx}{\textwidth}{lYYY}		\hline
	& \multicolumn{3}{c}{\textbf{Shock}} \\ \cline{2-4}
\textbf{Variable}~~~ & Demand & Oil-specific & Supply \\ \hline 
Activity	&	0.12	&	0.32	&	0.11	\\
Production	&	0.17	&	0.21	&	0.02	\\
Price	&	0.27	&	0.32	&	0.17	\\ 
\hline
	\end{tabularx}
    \par\vspace{0.5em}
\noindent\begin{minipage}{\textwidth}
\small\textbf{Notes:} Average of prior informativeness statistic over horizons. Higher values indicate that conditional prior has greater influence on posterior.
\end{minipage}
\end{table}

We compute this statistic for the impulse response of each variable to each shock at each horizon, then average it over horizons (Table~\ref{tab:priorinf}). Consistent with the discussion in Section 5.3, the prior informativeness statistic tends to be smaller for the responses to an oil supply shock than for the responses to other shocks. This indicates that the responses to an oil supply shock are more tightly identified, so the conditional prior contributes less of the apparent information in the posterior. The prior informativeness statistic averages less than 0.3 for almost all impulse responses, suggesting that the bulk of the information in the posterior is contributed by the data (given the identifying restrictions) rather than the conditional prior.

\renewcommand{\theequation}{S3-\arabic{equation}}
\renewcommand{\thesection}{S3}
\renewcommand{\thefigure}{S3-\arabic{figure}}
\setcounter{equation}{0}
\setcounter{figure}{0}
\section{Additional application: US monetary policy}
\label{app:monetarypolicy}

To explore whether the favourable performance of our approach persists in a larger model, we use the monetary SVAR from \citet{Antolin-Diaz_Rubio-Ramirez_2018} (henceforth, ARR18). They estimate the effects of monetary policy on the US economy using a six-variable SVAR in which the monetary policy shock is identified using an extensive set of sign restrictions on impulse responses and narrative restrictions.\footnote{See \citet{Giacomini_Kitagawa_Read_2023} for a robust Bayesian treatment of this application.}

The reduced-form VAR is the same as in \citet{Uhlig_2005}. The endogenous variables are real GDP ($GDP_t$), the GDP deflator ($GDPDEF_t$), a commodity price index ($COM_t$), total reserves ($TR_t$), non-borrowed reserves ($NBR_t$) (all in natural logarithms) and the federal funds rate ($FFR_t$); see \citet{Arias_Caldara_Rubio-Ramirez_2019} for details on the variables. The data are monthly and run from January 1965 to November 2007. The VAR includes a constant and 12~lags.

The sign restrictions on the impulse responses to a monetary policy shock follow \citet{Uhlig_2005}. The response of $FFR_{t+h}$ is non-negative for $h=0,1,\ldots,5$ and the responses of $GDPDEF_{t+h}$, $COM_{t+h}$ and $NBR_{t+h}$ are non-positive for $h=0,1,\ldots,5$.

We impose the extended set of narrative restrictions from Antol\'{i}n-D\'{i}az and Rubio-Ram\'{i}rez (2018)\nocite{Antolin-Diaz_Rubio-Ramirez_2018}. The restrictions are that the monetary policy shock was: positive in April 1974, October 1979, December~1988 and February 1994; negative in December 1990, October 1998, April 2001 and November 2002; and the most important contributor to the observed unexpected change in $FFR_t$ in these months. The implementation of the accept-reject sampler and our approach based on soft sign restrictions and the slice sampler is identical to that in the oil market application.

Obtaining the draws takes 21.8~hours using the accept-reject sampler, but only 1.1~hours using our approach. The average effective sample size from our sampler is 82.8~per cent, so our approach generates about 780,000~effective draws per hour, compared with around 45,000~effective draws per hour for the accept-reject sampler. Consistent with the results in Section~5.4 of the main text, our approach also more accurately determines whether the identified set is non-empty; the posterior plausibility of the identifying restrictions based on our sampler is 57.3~per cent compared with only 5.2~per cent when using the accept-reject sampler.

Figure~\ref{fig:impulseresponses_mp} summarises the posterior distributions of the impulse responses to a monetary policy shock obtained using the two samplers. The results are fairly similar, though the credible intervals based on the accept-reject sampler tend to be somewhat wider than those obtained using our sampler. This difference may reflect the ability of our sampler to better classify relatively small identified sets as non-empty. Figure~\ref{fig:posteriorcomp_mp} compares the approximation of the posterior distributions of the impact impulse responses obtained using the two samplers, conditional on the posterior mean of $\bm{\phi}$. The two approximations are again very similar. Overall, we conclude that our method continues to perform favourably in a larger SVAR.

\begin{figure}[h]
	\centering
	\caption{Impulse Responses to a Monetary Policy Shock -- Standard Bayesian Inference}
	\label{fig:impulseresponses_mp}
		\includegraphics[scale=1]{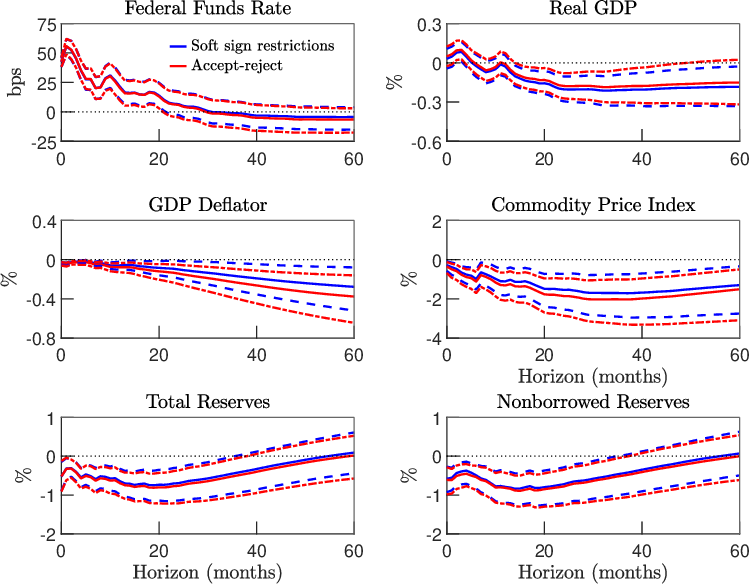}  
    \noindent\begin{minipage}{\textwidth}
	\small\textbf{Notes:} Solid lines are posterior medians and dashed lines are 68~per cent credible intervals.
    \end{minipage}
\end{figure}

\begin{figure}[h]
\centering
	\caption{Comparison of Posterior Approximations Across Samplers -- Impact Responses to a Monetary Policy Shock}
	\label{fig:posteriorcomp_mp}
	\begin{tabularx}{\textwidth}{X X X}
		\includegraphics[width=\linewidth]{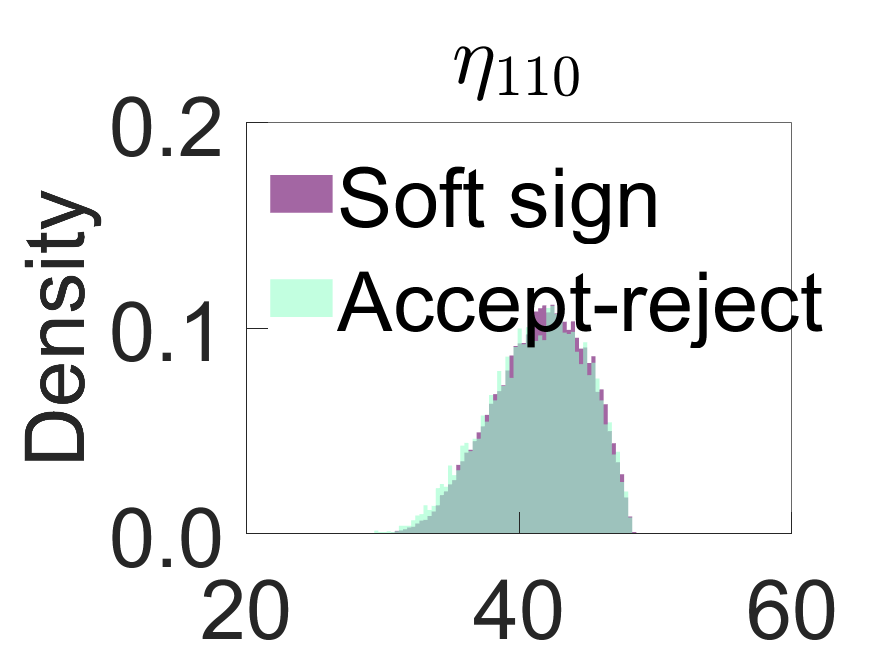}  &   \includegraphics[width=\linewidth]{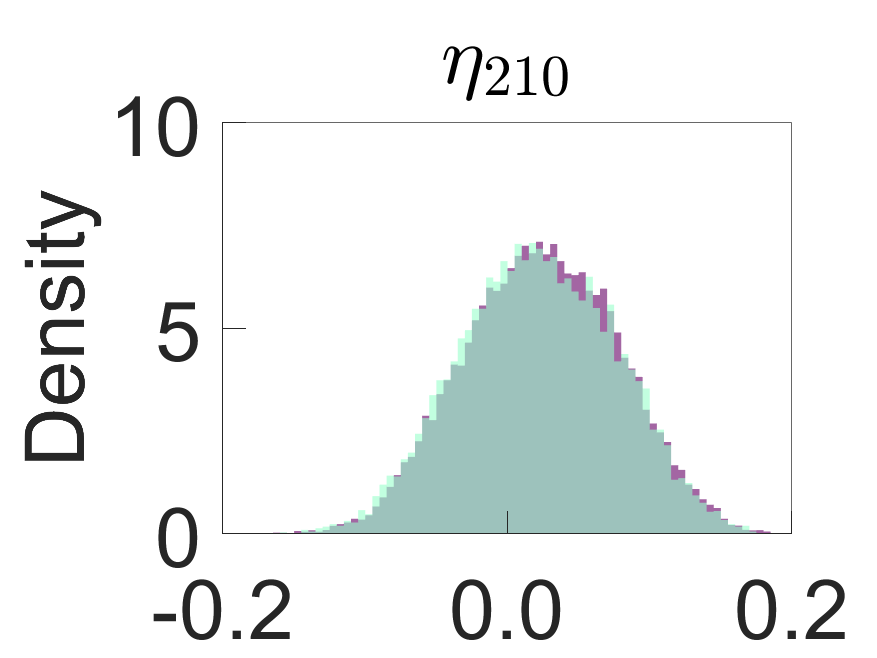}  & \includegraphics[width=\linewidth]{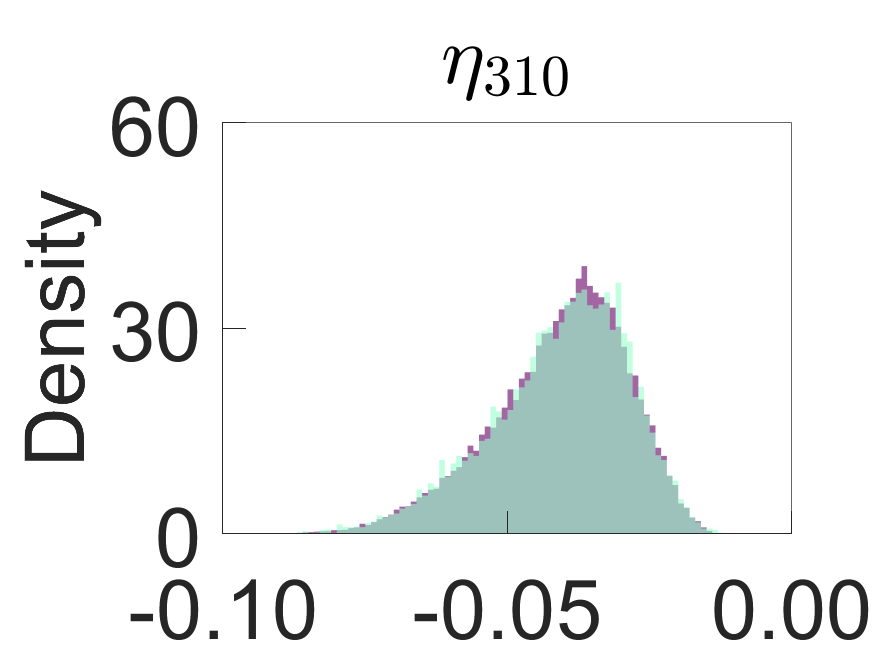}  \\
        \includegraphics[width=\linewidth]{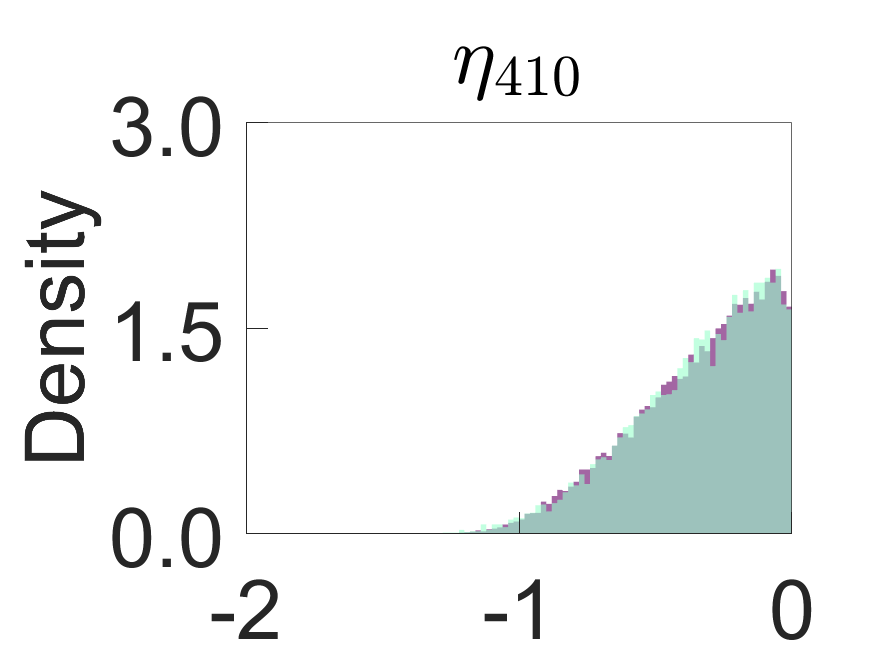}  &   \includegraphics[width=\linewidth]{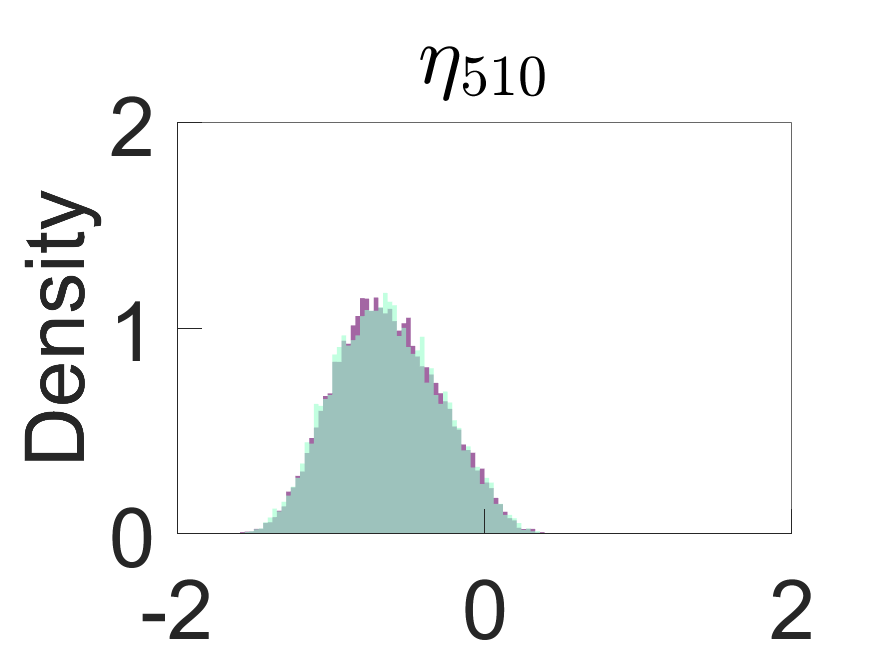}  & \includegraphics[width=\linewidth]{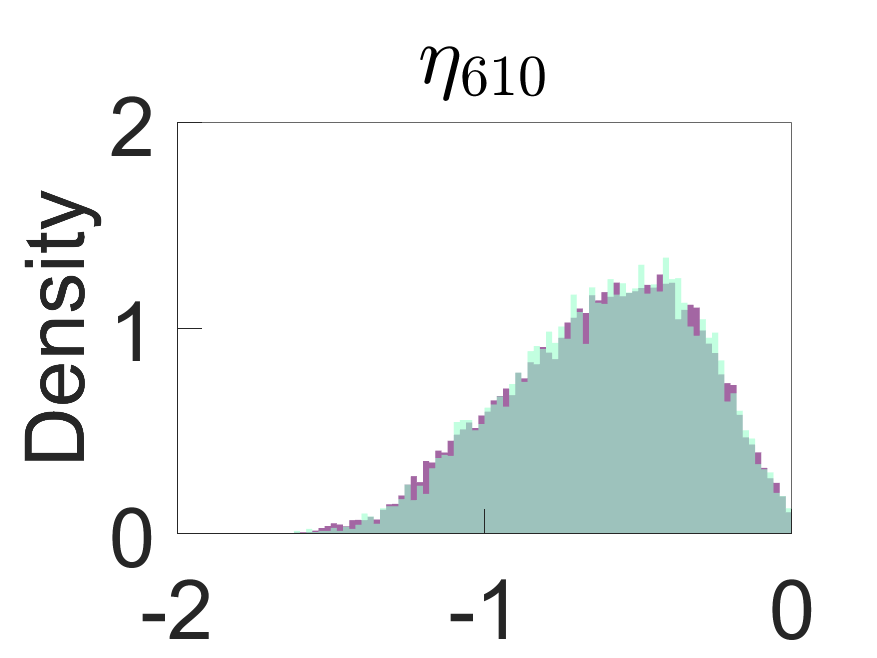}  \\
	\end{tabularx}
    \noindent\begin{minipage}{\textwidth}
	\small\textbf{Notes:} Based on 100,000~draws of $\mathbf{Q}$; conditional on posterior mean of $\bm{\phi}$.
     \end{minipage}  
 \end{figure}

\clearpage

\putbib
\end{bibunit}

\end{document}